\documentclass[10pt]{article}

\usepackage[margin=1in]{geometry}      
\usepackage{amsmath,amssymb,amsfonts,amsthm}  
\usepackage{graphicx}                  
\usepackage{authblk}                   
\usepackage[numbers,sort&compress]{natbib} 
\usepackage{hyperref}                  

\usepackage{url}
\usepackage{float}
\usepackage{setspace}
\usepackage{subcaption}
\usepackage{xcolor}         
\usepackage{overpic}
\usepackage{multirow}

\newtheorem{theorem}{Theorem}
\newtheorem{proposition}{Proposition}
\newtheorem{definition}{Definition}
\newtheorem{corollary}{Corollary}

\newcommand{\R}{\mathbb{R}}

\newcommand{\N}{\mathbb{N}}

\renewcommand{\eqref}[1]{(\ref{#1})}
\newcommand{\argmin}{\text{argmin}}

\newcommand{\revise}[1]{\textcolor{black}{#1}}
\newcommand{\issue}[1]{\textcolor{black}{#1}}

\hypersetup{
  colorlinks=true,
  linkcolor=black,
  citecolor=black,
  urlcolor=black,
  pdfauthor={},
  pdftitle={}
}

\title{\textbf{Learning Dissipative Chaotic Dynamics with\\ Boundedness Guarantees}}

\author{Sunbochen Tang}
\author{Themistoklis Sapsis}
\author{Navid Azizan}
\affil{Massachusetts Institute of Technology}

\date{} 

\begin{document}

\maketitle

\begin{abstract}
Chaotic dynamics, commonly seen in weather systems and fluid turbulence, are characterized by their sensitivity to initial conditions, which makes accurate prediction challenging. Recent approaches have focused on developing data-driven models that attempt to preserve invariant statistics over long horizons since many chaotic systems exhibit dissipative behaviors and ergodicity. \revise{Despite the recent progress in such models, they are still often prone to generating unbounded trajectories, leading to invalid statistics evaluation. To address this fundamental challenge, we introduce a modular framework that provides formal guarantees of trajectory boundedness for neural network chaotic dynamics models. Our core contribution is a dissipative projection layer that leverages control-theoretic principles to ensure the learned system is dissipative. Specifically, our framework simultaneously learns a dynamics emulator and an energy-like function, where the latter is used to construct an algebraic dissipative constraint within the projection layer. Furthermore, the learned invariant level set provides an outer estimate for the system's strange attractor, which is known to be difficult to characterize due to its complex geometry.} We demonstrate our model's ability to produce bounded long-horizon forecasts that preserve invariant statistics for chaotic dynamical systems, including Lorenz 96 and a reduced-order model of the Kuramoto-Sivashinsky equation.
\end{abstract}

\section{Introduction}
Chaos, characterized by exponential divergence after infinitesimal initial perturbations, is ubiquitous in a variety of complex dynamical systems, including climate models \citep{lorenz1963deterministic} and turbulence in fluids \citep{kuramoto1978diffusion,ashinsky1988nonlinear}. The exponential separation makes it challenging to accurately predict trajectories of chaotic systems. 
However, many chaotic systems of practical interest across various domains, including weather models and fluid dynamics \cite{lorenz1963deterministic, kuramoto1978diffusion}, turn out to be \emph{dissipative}, meaning that their trajectories converge to a bounded and positively invariant set, often referred to as a strange attractor \cite{stuart1998dynamical}. Moreover, trajectories of dissipative chaos will visit almost every state on the attractor, resulting in ergodicity and invariant statistics \cite{guckenheimer2013nonlinear}. 
Consequently, rather than seeking pointwise-accurate predictions, the primary goal in modeling dissipative chaotic systems becomes capturing these invariant statistics over long forecast horizons.

Recent data-driven efforts have shown remarkable progress in building surrogate models that accelerate inference while preserving the long-term invariant statistics of dissipative chaos. These methods span a wide spectrum of structural assumptions and model complexity. On one end, structured nonlinear regression introduces physically motivated multi-level models to fit time series data efficiently \citep{majda2001mathematical, majda2012physics}. At the other end, deep learning approaches rely on the representation power of neural networks to directly model complex chaotic behavior from raw data, while incorporating knowledge of shared physical system behaviors as specific architecture choices or regularization schemes \citep{li2020fourier, raissi2019physics, brunton2022data, lu2021learning, kochkov2021machine, page2024recurrent}. Hybrid approaches leverage autoencoder architectures to latent representation spaces where the dynamics evolve in simpler forms, with inspirations from Koopman theory \citep{koopman1931hamiltonian}, Dynamic mode decomposition \citep{kutz2016dynamic}, PCA \citep{pearson1901liii}, etc. Beyond one-step prediction, recurrent sequential models have also been explored to promote stability and improve forecast accuracy using more input information \cite{mikhaeil2022difficulty, vlachas2018data, sangiorgio2020robustness}. In addition to standard recurrent models, a specific recurrent network architecture design for time series prediction, known as reservoir computing (RC), has demonstrated improved performance in reconstructing attractors in chaos and preserving invariant statistics \citep{lu2018attractor, vlachas2020backpropagation, bollt2021explaining}.

To predict statistical properties on the attractor, data-driven models must generate arbitrarily long trajectories during inference to sufficiently sample the invariant measure. In practice, these models adopt an autoregressive paradigm that iteratively predicts the next state from its own prior outputs, making them vulnerable to drifting outside the region of training data. Since model performance is generally unpredictable outside of the training dataset, these models are prone to producing unbounded trajectories and invalid statistical forecasts. 
For models including quadratic regression and recurrent neural networks (RNNs), theoretical analysis has been established to reveal their fundamental difficulty in generating bounded trajectories for chaos forecasting \cite{majda2012fundamental, mikhaeil2022difficulty}. For more complex models, including RCs and FNOs, empirical evidence of the same stability issue has also been reported \cite{lu2018attractor, pathak2017using, li2022learning, schiff2024dyslim}. As such, a core challenge remains, which is preventing data-driven models from generating unbounded trajectories in long-term prediction for chaotic dynamical systems.

To address this challenge, we introduce a neural network framework that establishes formal guarantees of trajectory boundedness by enforcing dissipative constraints on the learned model. To the best of our knowledge, this is the first work to achieve such provable guarantees in neural network models for chaotic dynamics. Our core contribution is a novel projection layer that enforces dissipativity by construction, consequently establishing trajectory boundedness guarantees. Rather than relying on system-specific domain knowledge, our approach learns both the underlying dynamics and an energy function that governs the dissipative behavior directly from data. Additionally, our method results in a learned energy function invariant set, which provides a tight outer approximation of the strange attractor. We first illustrate the key components in our framework using the Lorenz 63 system, then demonstrate its effectiveness in preserving invariant statistics and generating bounded trajectories for Lorenz 96 and a reduced-order model of the Kuramoto–Sivashinsky (KS) equation, highlighting its broad applicability across different chaotic benchmarks.

In the recent literature, there have been several solutions proposed to address the stability issue, which generally fall into two categories: the first focuses on empirical stabilization, and the second seeks formal guarantees at the cost of limited model expressivity. Within the first category, some techniques are architecture-specific, such as the use of noise-inspired regularization \citep{wikner2022stabilizing} or network pruning \citep{haluszczynski2020reducing} to improve the long-term stability of RC models. Broader neural network-based approaches aim to preserve the system's invariant statistics by constraining models to match key dynamical properties like Lyapunov exponents \citep{platt2023constraining} or by adding statistics-based regularization to the loss function \citep{jiang2024training, schiff2024dyslim}, or by combining regularization and post-training modifications to encourage dissipative behaviors \citep{li2022learning}. While these methods reduce instances of unbounded prediction trajectories, they often require system-specific prior knowledge and, importantly, lack formal guarantees for trajectory boundedness. The second category of methods seeks to establish such guarantees by directly incorporating physical principles; however, so far, it has only been achieved for specific model classes. For instance, for quadratic regression models, energy conservation can be enforced structurally through a specific mathematical construction \citep{majda2012physics}. The explicit mathematical structure in such regression models makes it possible to embed energy conservation, but also limits their expressivity and prevents applications to more complex and realistic chaotic systems.

In this work, we design a general framework that aligns data-driven models with energy dissipation in chaos without reliance on explicit model structures, by formulating constraints based on learned energy representations. We simultaneously learn the flow of the system and an energy function directly from data, similar to prior work on learning stabilizing control inputs to unknown systems \cite{min2023data}. By leveraging Lyapunov stability theory, we derive efficient algebraic conditions that characterize energy dissipation purely based on the learned flow and energy function, thus removing the need for system-specific information. The learned energy function and control-theoretic algebraic conditions together allow the design of a closed-form dissipative projection layer, which automatically enforces the learned flow to be dissipative and consequently guarantees trajectory boundedness.

\begin{figure*}[h!]
    \centering
    \includegraphics[width=0.95\linewidth]{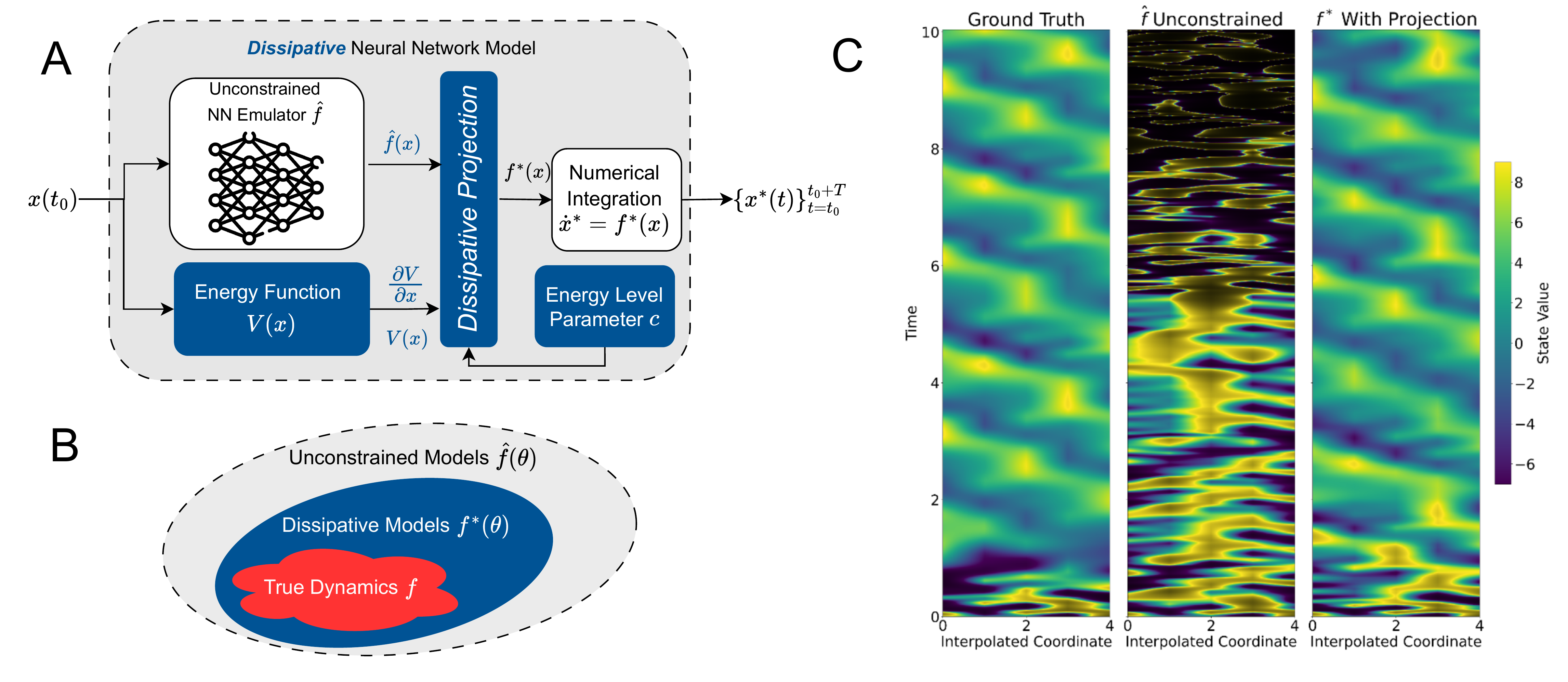}
    \caption{Overview of our approach and main results. (A) Current autoregressive models without constraints suffer from accumulated error, leading to trajectories growing unbounded over extended rollout. Our approach overcomes this issue by building a dissipative projection layer that ensures the model is dissipative and guarantees bounded trajectories. (B) The dissipative constraint effectively narrows the search space of model parameters because the true dynamics of interest are dissipative. (C) Hovm{\"o}ller diagrams of the Lorenz 96 system (with linear interpolation in the spatial dimension): Our approach ($f^*$ with projection, on the right) generates a bounded trajectory that reproduces flow characteristics seen in the ground truth trajectory (on the left), while the unconstrained neural network generates an unbounded trajectory failing to capture any meaningful statistics. }
    \label{fig:main}
\end{figure*}

\section{Data-driven Dissipative Chaos Modeling}
Consider a chaotic dynamical system described by the finite-dimensional ODE
\begin{equation}\label{eq:dynamics}
    \dot{x}(t) = f(x(t)), 
\end{equation}
where $x(t) \in \mathbb{R}^n$ represents the state of the dynamical system at time $t$, and $f: \R^n \to \R^n$ is a nonlinear function that governs the dynamics. Many chaotic systems are \textit{dissipative}, meaning that their trajectories eventually enter a bounded positively invariant set on which they exhibit chaotic behaviors. On this invariant set, also known as the strange attractor, the system becomes ergodic and consequently produces well-defined long-term statistical properties.
\begin{definition}\label{def:dissipativity}
    We say that the system in Eq.~\ref{eq:dynamics} is \textbf{dissipative} if there exists a bounded and positively invariant set $M \subset \mathbb{R}^n$ such that $$\lim_{t\to\infty} \text{dist}(x(t), M) = 0, \quad \text{dist}(x(t), M) = \inf_{y\in M} \|x(t) - y\|.$$
    In other words, every trajectory of the system will converge to $M$ asymptotically, and stays within $M$ once it enters. $M$ is said to be \textbf{globally asymptotically stable}.\footnote{The definition is derived based on the one introduced by \cite{stuart1998dynamical}, where they stated dissipativity as every trajectory will enter $M$ eventually. Here, we quantify this behavior using the notion of asymptotic stability and trajectory distance to $M$.}
\end{definition}

\issue{The practical significance of studying dissipative chaotic systems is the possibility to shift the focus from the challenging task of predicting exact chaotic trajectories to the more tractable goal of capturing their invariant statistical properties over long-term trajectory forecasts. In data-driven modeling for dissipative chaos, our objective becomes constructing a neural network dynamics emulator, $\hat{f}: \R^n \to \R^n$, that generates prediction trajectories $\{\hat{x}(t)\}_{t = t_0}^{t_0 + T}$ over a future time horizon $T$ by solving the initial value problem
\begin{equation}\label{eq:IVP}
    \dot{\hat{x}}(t)=\hat{f}\bigl(\hat{x}(t)\bigr), \quad \hat{x}(t_0)=x(t_0),
\end{equation}
from which the invariant statistics of the true system can be reproduced. The premise of reproducing meaningful statistical properties is that the trajectory remains bounded for a sufficiently long forecast horizon $T$. Although empirical success has been reported in a variety of recently proposed autoregressive models \cite{li2020fourier, lu2018attractor}, these models are prone to finite-time blowup in long trajectory rollouts, as illustrated in Fig.~\ref{fig:main}, which prohibits their use for statistical evaluation. This happens because the training dataset covers only a subset of the strange attractor, and without explicit guarantees, errors accumulate when the model encounters states outside the training set, leading to divergence. Moreover, the exponential separation characteristic of chaotic systems embedded in the training data further amplifies these errors, contributing to divergence in model rollouts during test time.
}

To address this challenge, our approach introduces a \textit{dissipative projection} layer to enforce an energy-based constraint into the neural network, which guarantees trajectories always remain bounded, ensuring their validity for long-term statistics evaluation. As illustrated in Fig.~\ref{fig:main}(A), our \textit{dissipative NN model} jointly learns a dynamics emulator, an energy function, and an energy level parameter directly from data. This approach forces the learned dynamics model $\dot{x}^* = f^*(x)$ to be dissipative and also provides an outer estimate for the strange attractor as an energy level set $\{x: V(x) \leq c\}$. In numerical experiments for the Lorenz 96 system \cite{lorenz1996predictability}, while we observe that the unconstrained model in \cite{li2020fourier} generated unbounded trajectories during test time, our constrained model produces bounded trajectories that match the statistics patterns seen in the ground truth Hovm{\"o}ller \cite{hovmoller1949trough} diagram. This result exemplifies the importance of enforcing dissipativity for reliable long-term predictions in dissipative chaos modeling. Beyond ensuring trajectory boundedness, as illustrated in Fig.~\ref{fig:main}B, since the true dynamics are dissipative, by only constructing models that are certified to be dissipative, this constraint effectively narrows the search space for neural network parameters. This is especially helpful in the limited data regime, where it becomes difficult for the model to learn the dissipative behaviors by simply fitting labeled data under trajectory prediction settings.


\section{Dissipative Dynamics: A Control-theory Perspective}
To develop models with inherent dissipativity, we first need to understand theoretical conditions that make a dynamical system dissipative. Note that dissipativity describes a system's energy behavior over time, independent of whether the system is chaotic or not. In this section, we focus on deriving algebraic conditions that are computationally efficient through the connection between dissipativity and energy, which are crucial for our proposed architecture that guarantees dissipativity.

\subsection{Energy-based Conditions for Dissipativity}

Dissipativity, as the name suggests, has a close relationship with energy in a dynamical system. Intuitively, a dissipative system will lose energy over time, which corresponds to the trajectory converging to a bounded set. Although the trajectory convergence behavior is quantitatively stated in Definition~\ref{def:dissipativity}, given a system $\dot{x} = f(x)$ known to be dissipative with access to the true dynamics $f$, it is still challenging to quantify the set $M$ that the trajectory converges to. In the specific context of dissipative chaos, there has been a body of literature that tries to address this issue by studying invariant manifolds, volume contraction, and attempting to characterize the strange attractor \citep{stuart1998dynamical}. Despite the rigorous treatment and the progress over the years that help us understand the strange attractor, these characterizations are often stated in abstract mathematical concepts and a descriptive manner that is intractable to computationally verify, e.g., \citep{milnor1985concept}. Given our goal of enforcing dissipativity in neural network models, it is crucial to first derive computationally efficient conditions that ensure a system is dissipative.

In control theory, the concept of Lyapunov functions has been used extensively to formalize asymptotic stability of dynamical systems, which are also known as ``energy-like'' functions due to strong connections with the mechanical energy of the system. More importantly, by leveraging the level set of such functions, numerous computationally tractable conditions have been derived and extensively used in designing practical controllers to ensure a system's asymptotic stability to equilibrium points \citep{khalil2002nonlinear}. By generalizing asymptotic stability with respect to an equilibrium point to a level set of a Lyapunov function, we derive computationally efficient conditions that ensure dissipativity in a dynamical system. 

Recall in Definition~\ref{def:dissipativity}, a dissipative system requires the existence of a bounded set $M$, which satisfies (1) $M$ is an invariant set (2) the system is globally asymptotically stable towards $M$. By reducing the definition to the existence of a Lyapunov function $V$ and choosing $M$ to be a level set of $V$, i.e., $M(c) = \{x: V(x) \leq c\}$ where $c > 0$ corresponds to the energy level, we derive the following conditions for invariance and asymptotic stability of $M(c)$ in Proposition~\ref{prop:invariance} and~\ref{prop:attractivity}, respectively.

\begin{proposition}[invariant level set]\label{prop:invariance}
    For a dynamical system in Eq.~\ref{eq:dynamics}, suppose there is a continuously differentiable scalar-valued function $V: \R^n \to \R$ and a constant $c > 0$, such that 
    \begin{align*}
        \forall x \in \{x \in \R^n: V(x) > c\}, \dot{V}(x) = \frac{\partial V}{\partial x} f(x) \leq 0. \footnotemark{}
    \end{align*}
    \footnotetext{Here $\frac{\partial V}{\partial x}$ refers to the row vector $[\frac{\partial V}{\partial x_1}, ..., \frac{\partial V}{\partial x_n}]$.}
    Then the level set $M(c) = \{x: V(x) \leq c\}$ is a positively invariant set for the system Eq.~\ref{eq:dynamics}.
\end{proposition}

\begin{proposition}[asymptotic stability]\label{prop:attractivity}
     For a dynamical system in Eq.~\ref{eq:dynamics}, suppose there is a lower-bounded continuously differentiable scalar-valued function $V: \R^n \to \R$ and a constant $c > 0$, such that 
    \begin{align*}
        &\text{ (1) } \forall x \in \{x \in \R^n: V(x) > c\}, \dot{V}(x) < 0;\\
        &\text{ (2) } V \text{ is radially unbounded}.
    \end{align*}
    Then the level set $M(c) = \{x: V(x) \leq c\}$ is globally asymptotically stable.
\end{proposition}

The proofs for both propositions are included in the SI Appendix. In Fig.~\ref{fig:invariance}, we sketch the intuitions behind the conditions in Proposition~\ref{prop:invariance} and~\ref{prop:attractivity}.

\begin{figure}[h!]
    \centering
    \includegraphics[width=\columnwidth]{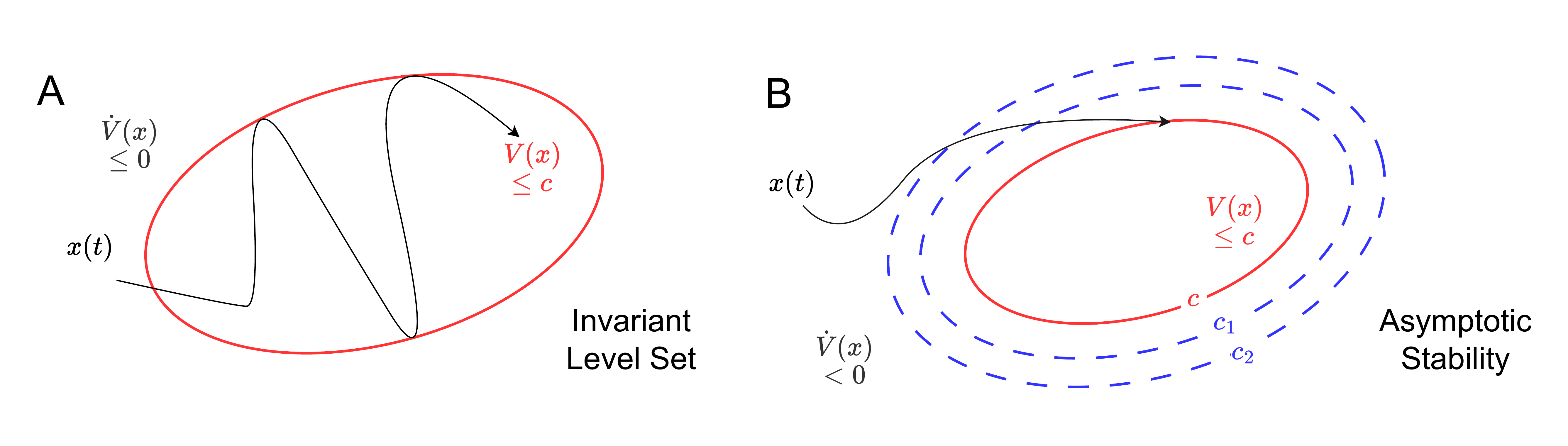}
    \caption{Illustrations of theoretical results in Proposition~\ref{prop:invariance}~and~\ref{prop:attractivity}: (A) The level set boundary serves as a barrier since the trajectory cannot gain energy outside. Once entering $M(c)$, the trajectory will be confined within. (B) The trajectory loses energy over time outside $M(c)$ because $\dot{V}(x) \leq 0$, resulting in convergence to the level set $M(c)$ eventually ($c_2 > c_1 > c > 0$).}
    \label{fig:invariance}
\end{figure}

\subsection{Algebraic Conditions for Dissipativity and Attractor Outer Estimation}

Compared to Definition~\ref{def:dissipativity}, the conditions derived in the propositions above are much more quantitative. However, verifying these conditions for a given system in Eq.~\ref{eq:dynamics} computationally is not trivial; for example, the negative semi-definite condition for $\dot{V}(x)$ is defined only on a certain part of the state space, outside the level set $M(c)$ in Proposition~\ref{prop:invariance}. Inspired by the S-procedure in sum-of-squares programming \citep{blekherman2012semidefinite}, we can replace these conditions with a single algebraic condition that is independent of the state, which is more computationally tractable.

\begin{theorem}\label{thm:main}
    Suppose there exists a lower-bounded radially unbounded $C^1$ function $V: \R^n \to \R$ and a constant $c > 0$ such that for the dynamical system in Eq.~\ref{eq:dynamics},
    \begin{equation}\label{eq:stability}
        \forall x \in \R^n, \dot{V}(x) + V(x) - c \leq 0.
    \end{equation}
    Then the system Eq.~\ref{eq:dynamics} is dissipative and $M(c)$ is globally asymptotically stable.
\end{theorem}
\begin{proof}
    The condition Eq.~\ref{eq:stability} implies that $\forall x \in \R^n$ such that $V(x) > c$, $\dot{V}(x) \leq -(V(x) - c) < 0$. Therefore, by Proposition~\ref{prop:invariance} and~\ref{prop:attractivity}, the level set $M(c)$ is both globally asymptotically stable and positively invariant.
\end{proof}

\textbf{Strange Attractor Outer Estimation.} Recall the characteristics of the strange attractor discussed in the background section; if only the first two properties are satisfied by a set $S$, then this set is known as ``the attractor'' \citep{strogatz2018nonlinear}. The difference from the strange attractor is that the attractor does not require the exponential separation of neighboring trajectories everywhere, hence it is also a superset of the strange attractor. Note that if a set $S\subset M(c)$ is an attractor, then $M(c)$ is both invariant and globally asymptotically stable. Therefore, in addition to certifying a system to be dissipative, the level set $M(c)$ also provides an outer approximation of the original attractor, which is consequently an outer estimate of the strange attractor as well.

In practice, characterizing the strange attractor is challenging due to its complex geometry \citep{milnor1985concept}. In addition to empirically reproducing the strange attractor using the learned model similar to \citep{li2022learning,jiang2024training}, our method also provides a level set outer estimate for the strange attractor by learning the Lyapunov function $V$ and level set parameter $c$.

\textbf{Importance of the Algebraic Condition.} The condition in Eq.~\ref{eq:stability} is strictly stronger than the conditions in Proposition~\ref{prop:invariance} and~\ref{prop:attractivity}. Although the condition is slightly more conservative, the computational tractability of obtaining a state-independent algebraic condition is well worth the trade-off. More importantly, this algebraic condition is crucial for constructing our proposed model architecture that ensures the learned system is dissipative, which will be discussed in detail in the next section.

\section{Inherently Dissipative Neural Network Dynamics Model}

Now we propose our architecture that simultaneously learns a dynamics emulator and a Lyapunov function $V$, with the former guaranteed to be dissipative by enforcing the condition in Eq.~\ref{eq:stability} through the construction of a projection layer. The overall structure is illustrated in Fig.~\ref{fig:main}(A).

\revise{Our proposed model consists of three learnable components: a neural network dynamics emulator $\hat{f}$ that approximates the true nonlinear dynamics $f$ on the right-hand side of Eq.~\ref{eq:dynamics}; a quadratic energy-like function $V(x) = (x - x_0)^T Q (x - x_0)$ serving as the Lyapunov function with learnable parameters $Q$ ($Q$ is parameterized to be positive definite) and $x_0$ that define the shape and center of the level sets, respectively; and a level set parameter $c > 0$. Using these three components and the gradient information $\partial V/\partial x$, we construct a ``dissipative projection'' layer that outputs a predicted dynamics emulator $f^*(x) \in \R^n$ which guarantees dissipativity of its corresponding dynamics $\dot{x} = f^*(x)$.}

\revise{A key feature of our method is its modularity, which provides the flexibility to enforce stability as an ``add-on" scheme. The proposed ``dissipative projection" layer is agnostic to the architectural choice for the dynamics emulator, which allows the integration of any suitable function approximator as the backbone model $\hat{f}$. Similarly, while we use a quadratic form for the Lyapunov function $V$ for simplicity, our framework accommodates any function that satisfies the conditions outlined in Theorem~\ref{thm:main}. This flexibility establishes our approach as a general-purpose pipeline for enforcing stability, rather than a solution dependent on a specific class of models.}


\subsection{Dissipative Projection Ensures Boundedness}

Following Theorem~\ref{thm:main}, if $\dot{x} = f^*(x)$ satisfies the condition in Eq.~\ref{eq:stability}, then the system is guaranteed to be dissipative and converge to the level set $M(c)$ defined by the learned Lyapunov function $V(x)$ and the learned constant $c > 0$. The dissipative projection layer is designed to modify the backbone emulator $\hat{f}$ such that the output $f^*$ produces a learned dynamical model $\dot{x} = f^*(x)$ that satisfies the condition in Eq.~\ref{eq:stability}, therefore certified to be dissipative.

Intuitively, this condition informs a subspace for the vector field $f(x)$ in which the forward dynamics will be dissipative. The dissipative projection layer is designed to project any unconstrained dynamics approximator, $\hat{f}(x)$, into such a subspace to ensure dissipativity.

More specifically, given an input $x \in \R^n$, the dissipative projection layer output $f^*(x)$ is chosen as the vector closest to the approximator $\hat{f}(x)$ under $l^2$ distance in the subspace of $\R^n$ defined by Eq.~\ref{eq:stability}, i.e., $f^*(x)$ is the solution to the following optimization problem:
\begin{subequations}
    \begin{align}
        f^*(x) =\ & \argmin_{f(x)} \|f(x) - \hat{f}(x)\|_2^2\quad \\
    & \text{subject to} \quad \frac{\partial V}{\partial x} f(x) + V(x) - c \leq 0
    \end{align}
\end{subequations}

As discussed extensively earlier, the constraint of the optimization problem, which is adopted from the condition in Eq.~\ref{eq:stability}, ensures the learned dynamics emulator $f^*(x)$ to be dissipative while being computationally efficient. In addition to the fact that the condition is easily verifiable through basic arithmetic operations, the constraint is also linear in the optimization variable $f(x)$. Since the above optimization problem has a quadratic loss and a linear constraint, an explicit solution can be found and computed using ReLU activation, similar to the approach in \cite{min2024hard,min2023data}:
\begin{equation}
    f^*(x)= \hat{f}(x) - {\frac{\partial V}{\partial x}}^T \frac{\text{ReLU}\left(\frac{\partial V}{\partial x} \hat{f}(x) + V(x) - c\right)}{\|\frac{\partial V}{\partial x}\|^2} \label{eq:projection}
\end{equation}

With the dissipative projection layer implemented as in Eq.~\ref{eq:projection}, we state the following corollary that formalizes the dissipativity of our proposed model architecture, which is a direct result of Theorem~\ref{thm:main}. A computational proof verifying $f^*(x)$ satisfies Eq.~\ref{eq:stability} $\forall x \in \R^n$ is included as well.
\begin{corollary}\label{cor:main}
    The learned dynamics model $\dot{x} = f^*(x)$ is a dissipative system with a bounded and positively invariant level set $M(c) = \{x: V(x) \leq c\}$. The set $M(c)$ is globally asymptotically stable, which implies every trajectory of the system is bounded and converges to $M(c)$ asymptotically.
\end{corollary}

\begin{proof}
    If $\frac{\partial V}{\partial x} \hat{f}(x) + V(x) - c\leq 0$, it follows from Eq.~\ref{eq:projection} that $f^*(x) = \hat{f}(x)$. Therefore, in this case, $\frac{\partial V}{\partial x} f^*(x) + V(x) - c \leq 0$, meaning that $f^*(x)$ satisfies the condition in Eq.~\ref{eq:stability}.

    If $\frac{\partial V}{\partial x} \hat{f}(x) + V(x) - c > 0$, then we have
    \begin{align*}
        \frac{\partial V}{\partial x} f^*(x) + V(x) - c
        &= \frac{\partial V}{\partial x} \left[\hat{f}(x) - \frac{\partial V}{\partial x}^T \frac{\text{ReLU}\left(\frac{\partial V}{\partial x} \hat{f}(x) + V(x) - c\right)}{\|\frac{\partial V}{\partial x}\|^2}\right] + V(x) - c\\
        &= \frac{\partial V}{\partial x} \hat{f}(x) - \left(\frac{\partial V}{\partial x} \hat{f}(x) + V(x) - c\right) + V(x) - c = 0,
    \end{align*}
    which ensures that $f^*(x)$ satisfies the stability condition Eq.~\ref{eq:stability}. Therefore, by Theorem~\ref{thm:main}, the system $\dot{x} = f^*(x)$ is dissipative and $M(c)$ is globally asymptotic stable.

    Additionally, since $\forall x \notin M(c)$, $\dot{V}(x) < -(V(x) -c) < 0$, the trajectory always loses energy outside $M(c)$. Also note that if the trajectory starts within $M(c)$, it can never leave, i.e., $V(x(t)) \leq c$ for all $t\geq 0$. Therefore, for any $t\in [0, \infty)$, $V(x(t)) = V(x(0)) + \int_{0}^{t} \dot{V}(x(\tau)) d\tau \leq \max\{V(x(0)), c\}$. Since every level set of $V(x(t))$ is bounded (see proof for Proposition~\ref{prop:attractivity} in SI Appendix), $x(t)$ is always bounded.
\end{proof}

\subsection{Training with Invariant Set Volume Regularization}

We consider a training dataset consisting of trajectory points which are evenly sampled at $h$ [sec] from a few ground truth trajectories initialized at randomly sampled initial conditions. Unlike \citep{li2022learning, jiang2024training}, we do not assume the trajectories in the training set are already inside the attractor, which allows for more flexibility when learning unknown chaotic systems, where the transition period before it reaches an invariant statistics is unknown. This is beneficial for our proposed model to learn where to place the invariant set $M(c)$ and apply dissipative projection. 

During training, we consider a multi-step setting, where we roll out the learned model for $T$ steps, each step sampled at $h$ [sec] using a numerical integration scheme. More specifically, given an initial condition chosen from the training dataset $x_0^{(i)}$, we forward simulate the learned system $\dot{\hat{x}}^{(i)} = f^*(\hat{x}^{(i)})$ with $\hat{x}^{(i)}(0) = x_0^{(i)}$ and obtain sampled states at the same sampling period $h$ [sec], $\hat{x}^{(i)}_k = \hat{x}^{(i)}(kh)$ at $k = 1, 2, ..., T$.  By sampling $N$ such trajectory snapshots of length $T$ from the training dataset, we define the prediction loss as the MSE between the predicted rollout sequence $(\hat{x}_k^{(i)})_{k=1}^T$ and the ground truth sequence $(x_k^{(i)})_{k=1}^T$: $\text{Prediction Loss} = \frac{1}{NT}\sum_{i=1}^N \sum_{k=1}^T \|x_k^{(i)} - \hat{x}_k^{(i)}\|_2^2$.

In the dissipative projection layer, the quadratic Lyapunov function $V(x)$ and the level set parameter $c$ both need to be optimized during training. Although the prediction loss depends on $V$ and $c$ through the projection operator that produces $f^*(x)$, optimizing only the prediction loss may not be a well-defined optimization problem. More specifically, if we have found a level set $M(c_1)$ that is globally asymptotically stable and invariant, then any superset $M(c_2)$ for $c_2 > c_1 > 0$ is also globally asymptotically stable and invariant. Therefore, there could potentially be infinitely many solutions that lead to the small prediction loss. 

To address this issue, we introduce a regularization loss that encourages the learned level set $M(c)$ to be as small as possible, which aligns with our goal of characterizing a tight outer estimate of the strange attractor. Toward this objective, we use the volume of the ellipsoid $M(c)$ as the regularization loss. Combining the prediction and regularization loss, we have the following training loss with a weight hyperparameter $\lambda > 0$ for balancing the regularization terms:
\begin{subequations}\label{eq:loss}
    \begin{align}
        &\text{Loss} = \frac{1}{NT}\sum_{i=1}^N \sum_{k=1}^T \|x_k^{(i)} - \hat{x}_k^{(i)}\|_2^2 + \lambda \text{Vol}(M(c)),\\
        &\text{Vol}(M(c)) = \frac{\pi^{n/2}}{\Gamma\left(\frac{n}{2} + 1\right)} \sqrt{\frac{c^n}{\det(Q)}}.
    \end{align}
\end{subequations}

\section{Results}

Unlike most data-driven methods that prioritize data over structures, our approach aligns the ML model with physics principles by learning both the dynamics and an energy representation and enforcing an energy dissipation constraint. We illustrate our method and demonstrate the effectiveness of providing formal trajectory boundedness guarantees in reliably reproducing invariant statistics through a set of numerical experiments based on Lorenz 63, Lorenz 96, and Kuramoto--Sivashinsky (KS) equation, with an increasing level of complexity. For more quantitative evaluations regarding invariant statistics and more visualizations, please refer to the SI Appendix.

\subsection{Learned flow is dissipative towards the invariant energy level set}

We first consider the classic Lorenz 63 system, originally proposed in \cite{lorenz1963deterministic} as a simplified model for atmospheric convection. The low-dimensionality of the system makes it easy to visualize both the strange attractor and dissipative behaviors of the system flow. As shown in Fig.~\ref{fig:L63}(A), a 50,000-step long-term trajectory rollout generated by our model (``fstar'') accurately reproduces the characteristic ``butterfly-shaped'' strange attractor of the ground truth (``GT''). The learned invariant energy level set, visualized as a yellow ellipsoid with a red dot indicating the learned center $x_0$, provides a tight outer-estimate for the strange attractor, and the trajectory quickly enters the level set and stays in afterwards.

Furthermore, Fig.~\ref{fig:L63}(B) illustrates the learned model's dissipative behavior by comparing its projected flow map on the  $x_1-x_2$ plane with the ground truth. The learned flow not only matches the ground truth near the attractor but also consistently points inward from outside the learned invariant set, confirming the model’s dissipativity by design. Overall, the Lorenz 63 example demonstrates our method's ability to learn 
geometrically interpretable and stable models, which guarantee to generate bounded trajectories reconstructing invariant statistics.

\begin{figure}[h!]
    \centering
    \includegraphics[width=0.6\columnwidth]{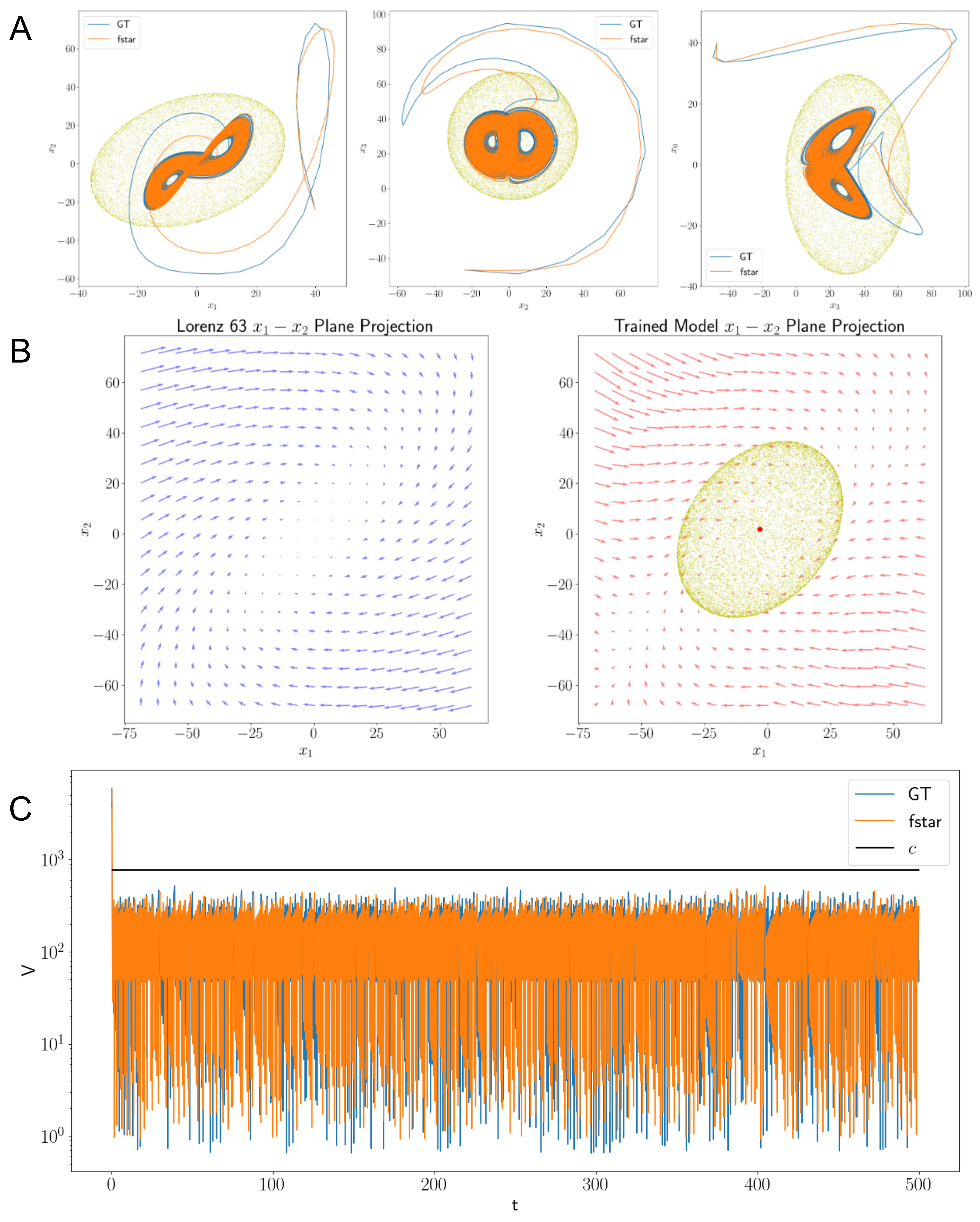}
    \caption{Lorenz 63 (A) Trajectories generated by the learned model (``fstar'') and true dynamics (``GT'') are visualized by their 2D projections, along with the learned invariant set (sampled with yellow points). (B) Comparison of the flow map projected onto $x_1-x_2$ plane. (C) Comparison of the learned energy function time history evaluated on the ground truth trajectory and the trajectory generated by our learned model.}
    \label{fig:L63}. 
\end{figure}

\subsection{Boundedness guarantee prevents finite-time blowup in high-dimensional systems}

To demonstrate the necessity of guaranteed boundedness in data-driven methods, we consider two higher-dimensional chaotic systems: the 5-dimensional Lorenz 96 system and a 32-dimensional reduced-order model of the KS equation (KS-ROM). For these more complex systems, unconstrained machine learning models experience finite-time blow-ups, resulting in prediction trajectories invalid for statistics evaluations.




As illustrated in Fig.~\ref{fig:main}(A), we benchmark against an unconstrained machine learning baseline which uses a standard multilayer perceptron (MLP) architecture to approximate the dynamics $f$ by $\hat{f}$. In the finite-dimensional case, the MLP baseline is equivalent to the Neural Operator \cite{kovachki2023neural} and the DeepONet \cite{lu2021learning}, both are state-of-the-art data-driven methods for general PDE modeling that do not incorporate any stability or energy constraints. For fair comparisons, we use the same MLP backbone for our method, augmenting it with an energy representation and the dissipative projection. For both Lorenz 96 and the reduced-order model for KS (KS-ROM), we train the unconstrained MLP and our proposed model on the exact same datasets, such that any difference in performance would be a direct consequence of the stability constraints. During testing, we sample a random initial condition, and forward simulate both the unconstrained model, $\dot{x} = \hat{f}(x)$, and our model with dissipative projection $\dot{x} = f^*(x)$ in an autoregressive manner (illustrated in Fig.~\ref{fig:main}(A)), and compare the trajectory rollouts generated by these two models. Since these system dimensions exceed three, we use Principal Component Analysis (PCA) to project the trajectory rollouts onto their first two principal components, such that a straightforward 2D visualization of the higher-dimensional trajectories can be presented.

The testing trajectories generated by the unconstrained model $\hat{f}$ completely deviate from the attractor and then exhibit finite--time blow-up for both Lorenz 96 and KS-ROM systems, as illustrated in Fig~\ref{fig:L96_KS_PCA}(A, C) respectively, by scatter plots comparing the first two PCA components of the generated trajectories with their corresponding ground truth trajectories. In contrast, our proposed model with dissipative projection guarantees trajectory boundedness, producing prediction trajectories that recover the shape of the strange attractor for both systems, as shown in Fig~\ref{fig:L96_KS_PCA}(A) and (C). Moreover, by applying the same PCA to sampled points on the boundary of the invariant level set $V(x) \leq c$, Fig.~\ref{fig:L96_KS_PCA}(A, C) not only verify the trajectory quickly converges to our learned level set (illustrated as the red point cloud), but also show that the level set form a tight outer-estimate of the strange attractor.

\begin{figure}[htp]
    \centering
    \includegraphics[width=\linewidth]{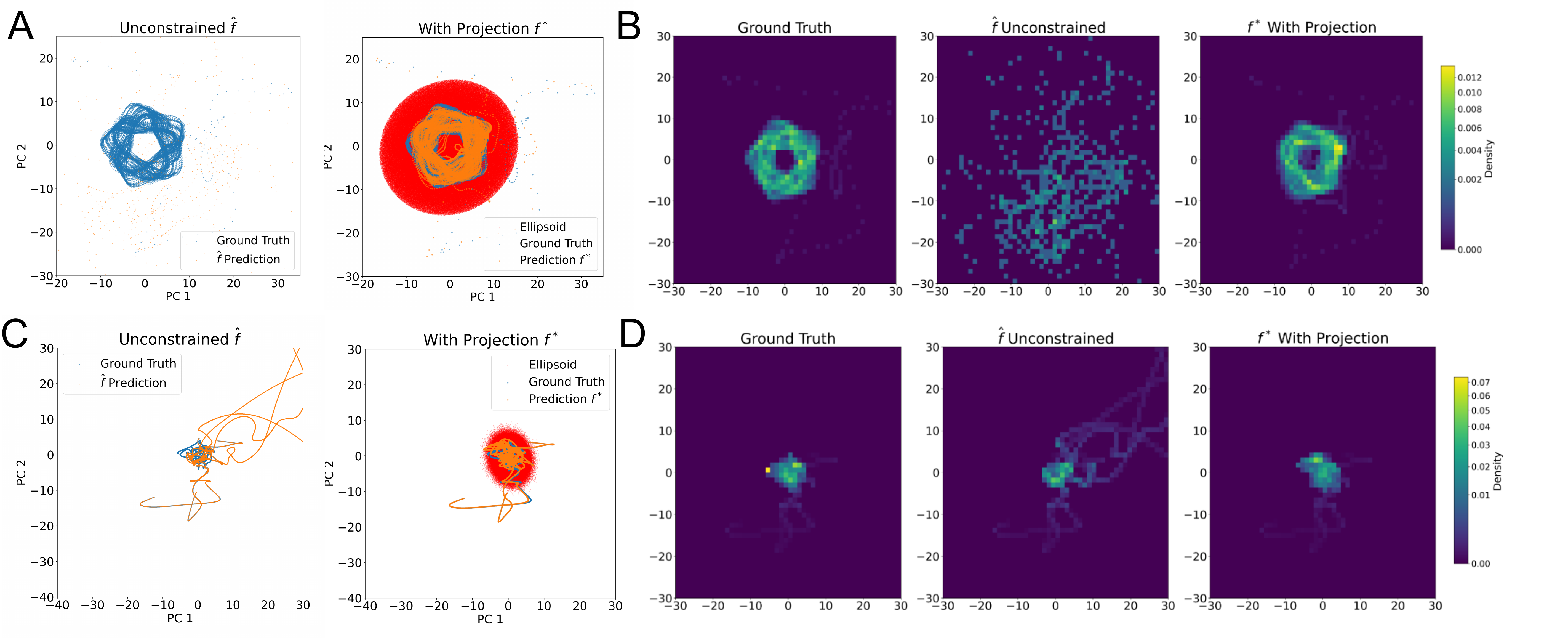}
    \caption{Principal component analysis (PCA) comparison between trajectories generated by unconstrained model ($\hat{f}$) and our proposed model with dissipative projection ($f^*$) for Lorenz 96 (first row, figures A and B) and a reduced-order KS model (second row, figures C and D). Unconstrained model (left figure in (A) and (C)) generates a trajectory that quickly deviates from the attractor and then grows unbounded. Our proposed model (right figure in (A) and (C)) provides boundedness guarantees, which guides the generated trajectory to enter and traverse the strange attractor. In addition, the learned invariant set (red ellipsoid point cloud) provides a tight outer-estimate of the attractor. The 2D histograms represent the probability density of trajectories in the PCA leading components for Lorenz 96 (B) and the reduced-order KS model (D). In both cases, the unconstrained model $\hat{f}$ (middle figure in (B) and (D)) produces an unstructured distribution that scatters across the PCA space, while our proposed model $f^*$ (right figure in (B) and (D)) is able to reproduce the shape and density of the ground truth distribution (left figure in (B) and (D)), which validates its capability to better preserve the system's invariant statistics.}
    \label{fig:L96_KS_PCA}
\end{figure}

\subsection{Enforcing boundedness helps preserve invariant statistics}

A key goal in modeling dissipative chaos is to reproduce the invariant statistics of the system's strange attractor, a task that requires generating bounded trajectories. When training data is limited, unconstrained models often fail by producing divergent trajectories, rendering statistical evaluation meaningless. Our method overcomes this by ensuring trajectory boundedness without prior knowledge of the system's statistics, unlike other approaches that incorporate known measures into the loss function \cite{jiang2024training, schiff2024dyslim}. We achieve this by constraining trajectories to a learned level set that acts as an outer estimate of the attractor, guiding the dynamics to the relevant region of the state space. To demonstrate the efficacy of this constraint, we compare rollouts from our model against an unconstrained baseline using: (1) spatiotemporal plots to illustrate flow patterns and (2) histograms of the leading principal components to show improved statistical distributions. Further visualizations, including Fourier spectra, Fourier modes, and learned energy characteristics, are provided in the SI Appendix.

\begin{figure}[h!]
    \centering
    \includegraphics[width=0.6\columnwidth]{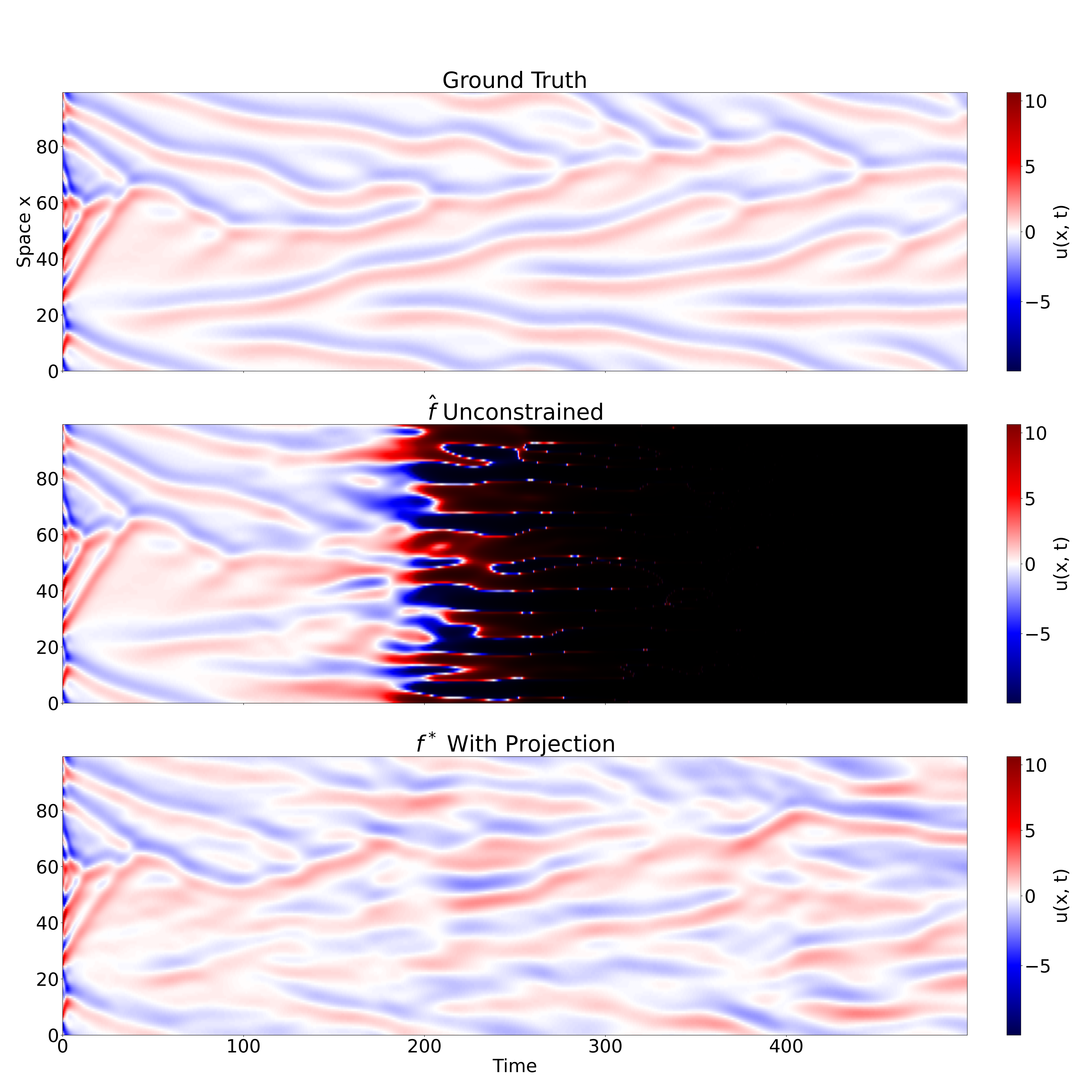}
    \caption{Reproduction of complex emergent dynamics in the Kuramoto-Sivashinsky system. The figure compares the spatiotemporal evolution of the ground truth (top) with predictions from the unconstrained ($\hat{f}$, middle) and our constrained ($f^*$, bottom) models. The unconstrained model exhibits finite-time blowup and grows unbounded. In contrast, our model's trajectory not only remains bounded but also faithfully captures the essential physical phenomena of the system, including the intricate patterns of coarsening events seen in the ground truth.}
    \label{fig:KS-sp}
\end{figure}

The spatiotemporal plots in Fig.~\ref{fig:main}(C) and Fig.~\ref{fig:KS-sp} validate that our model generates trajectories that recover system characteristics while the unconstrained baseline fails and diverges unboundedly. For the Lorenz 96 system, the Hovm\"{o}ller plot \cite{hovmoller1949trough} shows that the unconstrained model's solution diverges to unbounded values. Our model, however, maintains a stable rollout that correctly reproduces the westward-propagating wave patterns seen in the ground truth. Similarly, for the reduced-order model of the Kuramoto-Sivashinsky equation, the solution reconstructed from 32 Fourier modes shows a finite-time blowup for the unconstrained model. In contrast, our model's trajectory remains bounded and successfully captures the complex emergence of cellular chaos and coarsening events characteristic of the true dynamics.

To quantitatively evaluate the preservation of invariant statistics, we project the long-horizon testing trajectory rollouts onto their first two principal components (PCs) and compute the probability density function (PDF) of the projected points. These leading PCs capture dominant modes of variance and provide a low-dimensional visualization of the geometric structure of the attractor. The 2D histograms in Fig.~\ref{fig:L96_KS_PCA}(B, D) compare the PDFs of the ground truth (left), unconstrained model ($\hat{f}$, middle), and our proposed model ($f^*$, right). In the region near the attractor, the unconstrained model's trajectory distribution is scattered and unstructured, failing to capture the distinct characteristics of the ground truth attractor. In contrast, our model not only confines the trajectory to the correct region of the subspace but also reproduces the intricate shape of the invariant measure.

We provide a quantitative comparison of the statistical property prediction errors between the unconstrained model $\hat{f}$ and our constrained model $f^*$ for both Lorenz 96 and KS-ROM systems. First, the Kullback-Leibler (KL) divergence between the predicted and the true PCA distributions further validates the visual assessment in Fig.~\ref{fig:L96_KS_PCA}(B, D): our model's divergence (Lorenz~96: 0.4078; KS-ROM: 0.8355) is significantly lower than that of the unconstrained model (Lorenz~96: 9.7856, KS-ROM: 3.0557) which fails to capture the attractor's geometric characteristics. 
Second, we evaluated the relative error of the Fourier energy spectra acquired from the predicted trajectory rollout compared to ground truth. Here again, the unconstrained model exhibits orders of magnitude higher error due to finite-time blowup (Lorenz~96: 164.10; KS-ROM: 244.47). In stark contrast, our proposed model's spectrum matches the ground truth closely (Lorenz~96: 0.0213; KS-ROM: 0.0001). These metrics demonstrate that our method effectively preserves the invariant statistics of chaotic systems by ensuring trajectory boundedness through dissipative constraints, even when trained on limited data and without prior knowledge of the system's statistical properties.

\section{Discussion}

In this paper, we have proposed a general framework that provides formal guarantees of trajectory boundedness in data-driven chaotic dynamics models, which directly addresses a crucial stability issue in reliable applications of these models to long-horizon forecasts. To achieve such guarantees without heavy reliance on system-specific domain knowledge, we derived efficient algebraic conditions characterizing dissipative dynamics based on analyzing an energy function, by leveraging stability results in control theory. We designed a dissipative projection layer to enforce such dissipative algebraic conditions by construction in our model, which learns the dynamics and the energy function concurrently from data. Additionally, we incorporated an invariant set volume-based regularization, which enables the model to self-discover the energy constraint and an invariant set that serves as a tight outer approximation for the strange attractor. As such, our framework ensures that the learned neural network model is provably dissipative, thus providing trajectory boundedness guarantees. Numerical experiments, including on Lorenz~96 and a reduced-order KS model, demonstrated that our model generates bounded trajectories and, in turn, reconstructs invariant statistics that closely match those of the true dynamics.



Notably, we also observed that an unconstrained model without the projection layer, trained on the exact same dataset, generated unbounded trajectories. This ablation revealed both the effectiveness and the necessity of establishing formal guarantees in machine learning models. For physical systems where collecting data can be expensive, a low-data regime, where an unconstrained model cannot learn stable behaviors purely from the given dataset, can become common. Our proposed projection layer suggests that the model quality can be significantly improved by adding the correct inductive bias devised from physical principles, without the need for any additional training data.


In summary, we advance data‑driven modeling of chaotic dynamical systems by establishing formal guarantees of trajectory boundedness. By embedding hard constraints as an inductive bias, our models are reliable by construction, consistent with known physical principles, without requiring additional data or compute. This capability is especially valuable in data‑scarce or safety-critical scientific domains, such as climate modeling, optimal design of fluid systems, and autonomous systems. More broadly, our framework for constraining neural‑network outputs mitigates the black‑box nature of ML models and their unpredictable behavior at deployment. Looking ahead, we plan to extend these guarantees to modeling more complex real-world systems and their downstream control tasks. We expect this line of work to inform model design and enable more trustworthy applications in complex cyber-physical systems.

\section*{Materials and Methods}

\subsection*{Numerical Models of Chaotic Systems}
We used three chaotic systems of increasing complexity for our numerical experiments. 
Ground truth data for all systems was generated by integrating the governing equations 
using the 4th-order Runge-Kutta (RK4) method with a time step of $h = 0.01$ seconds.

\paragraph{Lorenz 63 System.}
The Lorenz 63 system \cite{lorenz1963deterministic} is a 3-dimensional ordinary differential 
equation (ODE) given by:
\begin{align*}
    \dot{x}_1 &= \sigma (x_2 - x_1), \\
    \dot{x}_2 &= x_1(\rho - x_3) - x_2, \\ 
    \dot{x}_3 &= x_1 x_2 - \beta x_3,
\end{align*}
where the state is $x \in \mathbb{R}^3$. We used the standard chaotic parameters 
$\sigma=10.0$, $\beta=8/3$, and $\rho=28.0$.

\paragraph{Lorenz 96 System.}
The Lorenz 96 system \cite{lorenz1996predictability} is a higher-dimensional ODE. We used a 
5-dimensional version described as follows, with the forcing constant set to $F = 8.0$ to generate chaotic behavior.
\begin{align*}
    \dot{x}_i &= (x_{i+1} - x_{i-2})x_{i-1} - x_i + F, \quad i = 1, \dots, 5, \\
    \text{where } & x_{-1} = x_4, x_{0} = x_5, x_6 = x_1.
\end{align*}

\paragraph{Kuramoto--Sivashinsky (KS) Reduced-Order Model (ROM).}
We used a 32-dimensional ODE model derived from the Kuramoto--Sivashinsky (KS) partial 
differential equation \cite{kuramoto1978diffusion}. The reduction was performed using the 
Galerkin projection method \cite{smyrlis1996computational, holmes2012turbulence}. Full details 
on the derivation are provided in the SI Appendix.

\subsection*{Model Architecture and Training}

\paragraph{Model Architecture.}
Our proposed model, $f^*$, and the unconstrained baseline model, $\hat{f}$, share an 
identical multilayer perceptron (MLP) backbone to ensure a fair comparison. 

The MLP backbone $\hat{f}$ consists of two hidden layers and uses a GeLU \cite{hendrycks2016gaussian} activation function for all three systems. For the Lorenz 63 and 96 systems, each hidden layer contains 64 neurons. For the reduced-order KS (KS-ROM) system, the hidden layer width increases to 512 neurons as the dynamics the model tries to capture become much more complex.

Our model, $f^*$, augments this backbone with an additional network to represent the energy function $V(x)$ and a final projection layer that enforces the energy dissipation constraint. The energy function $V(x)$ is parameterized as a quadratic function $V(x) = (x-x_0)^T Q (x-x_0)$, where the positive-definite matrix $Q$ and center $x_0$ are learned. Note that $Q$ is constructed to be positive definite, as the learnable parameters are log-diagonal terms and off-diagonal terms of a lower triangular matrix $L$ forming the Cholesky decomposition $Q = L L^T$.

\paragraph{Training Datasets and Procedure.}
All models were trained using the Adam optimizer with a learning rate of $1\times 10^{-4}$ for 30,000 epochs. For our proposed model, the level set parameter $c$ is initialized at a sufficiently large number, and the regularization weight parameter $\lambda$ can be chosen over a hyperparameter grid search over $10^{-1}, 10^{-2}, ..., 10^{-8}$.

The baseline MLP and our constrained model were trained on the exact same datasets for each system:
\begin{itemize}
    \item \textbf{Lorenz 63:} The dataset consisted of 10 trajectories, each with 200 steps. 
    Training was performed using a multi-step prediction loss with a horizon of $T=5$.
    \item \textbf{Lorenz 96:} The dataset for each system consisted of 4 trajectories, 
    each with 500 steps. Training was performed using a single-step prediction loss.
    \item \textbf{Reduced-order KS:} The dataset for each system consisted of 4 trajectories, 
    each with 50,000 steps. Training was performed using a multi-step prediction loss with a horizon of $T=5$.
\end{itemize}

\subsection*{Evaluation and Visualization}

\paragraph{Trajectory Rollouts and Simulation.}
For testing, long-horizon trajectories were generated by simulating the learned models 
autoregressively from a random initial condition not seen during training. Trajectory 
integration was performed using the RK4 method with a time step of $h = 0.01$ seconds. 
For the Lorenz 63 and the KS-ROM system, 10 trajectories of 50,000 steps were generated for creating visualizations and evaluating the statistical properties of the learned systems. For Lorenz 96, 20 trajectories of 20,000 steps were generated for evaluation.

\paragraph{Statistical Metrics.}
We used two quantitative metrics to evaluate the models' ability to reproduce the invariant 
statistics of the true system from the long-horizon rollouts.
\begin{itemize}
    \item \textbf{KL Divergence:} The Kullback-Leibler (KL) divergence was calculated between the 
    2D probability density function (PDF) of the model's trajectory projected onto its first 
    two principal components and the PDF of the ground truth trajectory. The PDFs were estimated 
    using 2D histograms with $50\times 50$ bins.
    \item \textbf{Spectrum Error:} The Fourier energy spectrum was first computed for the long-horizon trajectory rollouts. 
    To quantify the discrepancy, we calculated the relative L2 norm error between the predicted 
    spectrum ($S_{\text{pred}}$) and the ground truth spectrum ($S_{\text{true}}$). This error is given by the formula:
    $$
    \text{Error} = \frac{\| S_{\text{pred}} - S_{\text{true}} \|_2}{\| S_{\text{true}} \|_2},
    $$
    where $\| \cdot \|_2$ denotes the standard Euclidean norm. This metric measures the 
    overall deviation of the predicted spectrum's shape and magnitude relative to that of the 
    ground truth.
\end{itemize}

\section*{Acknowledgments}
The authors acknowledge the MIT SuperCloud and Lincoln Laboratory Supercomputing Center for providing computing resources that have contributed to the results reported within this paper. This work was supported in part by MathWorks, the MIT-Amazon Science Hub, the MIT-Google Program for Computing Innovation, and the MIT-IBM Watson AI Lab.

\bibliographystyle{unsrtnat} 
\bibliography{ref}          


\newpage
\appendix

\section{Proof for Theoretical Results}
\begin{proof}[Proof for Proposition~\ref{prop:invariance}]
    Let $x(0) \in M(c)$, suppose there exists $t > 0$ such that $x(t_0) \notin M(c)$, which implies that $V(x(0)) \leq c < V(x(t_0))$. Since $V(x(t))$ is continuously differentiable in $t$, by the intermediate value theorem, we can find some $t\in [0, t_0)$ such that $V(x(t_0)) = c$, denote $S = \{t: V(x(t)) = c, t \in [0, t_0)\}$. Since $\forall t \in S, t < t_0$, $\sup S \leq t_0$. Suppose $\sup S = t_0$, then we can construct a sequence $(t_k)_{k\in \N}$ such that $t_k \to \sup S$ as $k\to \infty$. By continuity, $V(x(\sup S)) = c$ which contradicts $V(x(t_0)) > c$. Therefore, $\sup S < t_0$. Now by the mean value theorem, there exists $t_1 \in (\sup S, t_0)$ such that $V(x(t_1)) > c$ and $\dot{V}(x(t_1)) = (V(x(t_0)) - V(x(\sup S)))/(t_0 - \sup S) > 0$, which contradicts the assumed condition. Therefore, $\forall t > 0$, we have $V(x(t)) \in M(c)$ as long as $x(0) \in M(c)$, i.e., $M(c)$ is indeed a positively invariant set.
\end{proof}

\begin{proof}[Proof for Proposition~\ref{prop:attractivity}]
    Since the condition here is stronger than the one stated in Proposition 1, $M(c)$ is a positively invariant set. Therefore, it suffices to consider a trajectory that starts outside $M(c)$. Suppose there exists a trajectory $x(t)$ such that $\forall t \in [0, \infty), V(x(t)) > c$, then $\dot{V}(x(t)) < 0$ and $V(x(t))$ is monotonically decreasing over time. Since $V(x(t))$ is lower bounded, $V(x(t)) \to a\geq \inf_x V(x)$ as $t\to \infty$. Suppose $a > c$, i.e., $\text{dist}(x(t), M(c)) = \inf_{y\in M(c)} \|y - x(t)\| \not\to 0$ as $t\to\infty$.

    Since $V$ is radially unbounded, for any $\alpha > 0$, we can find $r_{\alpha}$ such that $V(x) > \alpha$ for all $\|x\| > r_{\alpha}$. Therefore, any level set of $V$ is bounded as $\{x: V(x) \leq \alpha\} \subset B(r_{\alpha})$. Note that $V(x(t)) \in [a, V(x(0))]$ for all $t\in [0, \infty)$ because $V(x(t))$ is monotonically decreasing. Since $V(x)$ is continuous, the pre-image $S = \{x: V(x) \in [a, V(x(0))]$ is a closed set. Additionally, $S$ is bounded because $S \subset \{x: V(x) \leq \alpha\} \subset B(r_{V(x(0))})$, which implies $S$ is compact.
    
    Since $\dot{V}(x)$ is continuous and $\dot{V}(x) < 0$ for all $x \in S$, there is $\gamma > 0$ such that $\max_{x\in S}\dot{V}(x) \leq -\gamma$, which implies $\max_{t\in [0, \infty)}\dot{V}(x(t)) \leq \max_{\|x\|<r_{x(0)}}\dot{V}(x) \leq -\gamma$. This contradicts the fact that ${V}(x(t)) \geq a > -\infty$ for all $t\in [0, \infty)$ since
    \begin{equation*}
        V(x(t)) = V(x(0)) + \int_0^t \dot{V}(x(\tau)) d\tau \leq V(x(0)) - \gamma t.
    \end{equation*}
\end{proof}

\section{Reduced-order Kuramoto-Sivashinsky Model}

We consider the one-dimensional Kuramoto--Sivashinsky (KS) equation on a
periodic domain of length $L$,
\begin{subequations}
    \begin{align}
        u_t(x,t) &= -u_{xx}(x,t) - \nu\,u_{xxxx}(x,t)
                   - \frac{1}{2}\bigl(u(x,t)^2\bigr)_x, \label{eq:KS_phys}\\
        (x,t) &\in [0,L]\times\mathbb{R}^+,\\
        u(x+L,t) &= u(x,t), \label{eq:KS_phys_bc}
    \end{align}
\end{subequations}
where $\nu>0$ is a hyperviscosity parameter. For $L=100$ and $\nu=4$, Eq.~\ref{eq:KS_phys}
is in a strongly chaotic regime, consistent with the settings used in previous
spectral studies of the KS equation
\citep{smyrlis1996computational,trefethen2000spectral,boyd2001chebyshev}.

Following \citep{smyrlis1996computational} and the general Galerkin projection framework
for dissipative PDEs \citep{holmes2012turbulence}, we construct a truncated
reduced-order model of Eq.~\ref{eq:KS_phys} on the periodic domain. Using the complex Fourier basis
$\{\mathrm{e}^{\mathrm{i}k_n x}\}$, we expand
\begin{equation}\label{eq:KS_fourier_expansion}
    u(x,t) = \sum_{n=-\infty}^{\infty} \hat u_n(t)\,\mathrm{e}^{\mathrm{i}k_n x}, \quad k_n = \frac{2\pi n}{L},\qquad n\in\mathbb{Z}
\end{equation}
with complex Fourier coefficients $\hat u_n(t)\in\mathbb{C}$. Since the solution operator $u$ always has real values, the Hermitian symmetry $\hat u_{-n}(t) = \overline{\hat u_n(t)}$ needs to hold. Due to the periodic boundary conditions, 
the energy conservation suggests that $\hat{u}_0(t)$ will be constant in time. Without loss of generality, we assume $\hat{u}_0(t) = 0$, similar to the treatment in \citep{smyrlis1996computational}.

Using partial derivatives
\(
    \partial_{xx} \mathrm{e}^{\mathrm{i}k_n x} = -k_n^2 \mathrm{e}^{\mathrm{i}k_n x}, 
    \partial_{xxxx} \mathrm{e}^{\mathrm{i}k_n x} = k_n^4 \mathrm{e}^{\mathrm{i}k_n x},
\)
the linear part of Eq.~\ref{eq:KS_phys} becomes
\[
    -u_{xx} - \nu u_{xxxx}
    = \sum_{n=-\infty}^{\infty} \bigl(k_n^2 - \nu k_n^4\bigr)\,\hat u_n(t)\,
      \mathrm{e}^{\mathrm{i}k_n x}.
\]
For the nonlinear quadratic term $\bigl(u(x,t)^2\bigr)_x$, we first expand the Fourier series for
\[
    u(x,t)^2
    = \sum_{p=-\infty}^{\infty} \sum_{q=-\infty}^{\infty}
      \hat u_p(t)\,\hat u_q(t)\,\mathrm{e}^{\mathrm{i}(k_p+k_q)x}
    = \sum_{n=-\infty}^{\infty} \widehat{u^2}_n(t)\,\mathrm{e}^{\mathrm{i}k_n x},
\]
where the Fourier coefficients of $u^2$ are given by the convolution
\begin{equation}\label{eq:KS_conv_full}
    \widehat{u^2}_n(t)
    = \sum_{p+q=n} \hat u_p(t)\,\hat u_q(t),\qquad n\in\mathbb{Z}.
\end{equation}
Differentiating in $x$, we have the nonlinear term rewritten as
\[
    -\frac{1}{2}\bigl(u(x,t)^2\bigr)_x
    = \sum_{n=-\infty}^{\infty}
      \left( -\frac{\mathrm{i}k_n}{2}\,\widehat{u^2}_n(t)\right)
      \mathrm{e}^{\mathrm{i}k_n x}.
\]
The above derivation transforms the PDE in Eq.~\ref{eq:KS_phys} into the following infinite-dimensional ODEs,
\begin{equation}\label{eq:KS_modal_infinite}
    \frac{\mathrm{d}}{\mathrm{d}t}\hat u_n(t)
    = \bigl(k_n^2 - \nu k_n^4\bigr)\,\hat u_n(t)
      - \frac{\mathrm{i}k_n}{2}\,\widehat{u^2}_n(t),
    \qquad n\in\mathbb{Z},\ \hat u_0 = 0,
\end{equation}
which is a standard Galerkin projection of the KS equation on a periodic domain
\citep{smyrlis1996computational,trefethen2000spectral,boyd2001chebyshev}.

To obtain a finite-dimensional reduced-order model, we truncate the expansion
Eq.~\ref{eq:KS_fourier_expansion} to modes $|n|\le M$. Since the basis is chosen to be orthonormal, the truncated Fourier series satisfies the following finite-dimensional system of complex ODEs,
\begin{subequations}\label{eq:KS_modal_truncated}
    \begin{align}
        \frac{\mathrm{d}}{\mathrm{d}t}\hat u_n(t)
        &= \bigl(k_n^2 - \nu k_n^4\bigr)\,\hat u_n(t)
           - \frac{\mathrm{i}k_n}{2}
             \sum_{\substack{p+q=n\\ |p|,|q|\le M}} \hat u_p(t)\,\hat u_q(t),
        \\
        &\qquad\qquad\qquad\quad n = -M,\dots,-1,1,\dots,M.
    \end{align}
\end{subequations}

To simplify the dataset for our learning framework, we redefine the state vector as a vector of real-valued entries. Specifically, every Fourier coefficient in Eq.~\ref{eq:KS_modal_infinite} can be written as
\(
    \hat u_n(t) = a_n(t) + \mathrm{i} b_n(t), a_n(t),b_n(t)\in\mathbb{R}.
\)
Using the symmetry $\hat u_{-n} = \overline{\hat u_n}$, we can represent all $2M$ complex coefficients by the set of functions $\{a_n(t), b_n(t)\}_{n=1}^M$. We collect these into a $2M$--dimensional real state,
\begin{equation}\label{eq:KS_ROM_state}
    x(t)
    = \bigl(a_1(t),\dots,a_M(t),\ b_1(t),\dots,b_M(t)\bigr)^\top
    \in \mathbb{R}^{2M},
\end{equation}
and denote the resulting truncated ODE by
\begin{equation}\label{eq:KS_ROM_ode}
    \dot x(t) = f_{\mathrm{KS}}(x(t)).
\end{equation}
In numerical experiments we choose $M=16$ and form a $32$-dimensional truncated ODE. This low-dimensional model retains the chaotic behavior of the KS dynamics with the chosen scale and viscosity parameter, as discussed in previous studies of low-dimensional KS models \citep{smyrlis1996computational,holmes2012turbulence}.

In practice, the convolution in Eq.~\ref{eq:KS_conv_full} is evaluated pseudospectrally: at each time step we reconstruct $u(x,t)$ on an even-distance grid of $N_{\text{grid}}$ collocation points via an inverse FFT, compute $u(x,t)^2$ pointwise in physical space, and apply a forward FFT to obtain $\widehat{u^2}_n(t)$ \citep{trefethen2000spectral,boyd2001chebyshev}. To
control aliasing errors from the quadratic nonlinearity, we follow the $2/3$--rule dealiasing by choosing $N_{\text{grid}}\ge 3M$ and zeroing the highest one-third of modes in the FFT of $u^2$ before reading off $\widehat{u^2}_n$ for $|n|\le M$.

\section{Additional Numerical Results}\label{sec:append_results}

\subsection{Spatial Fourier spectra}\label{sec:append_stats_eval}
In the main text, we presented PCA-based analyses (Figure~4) that demonstrate our model's ability to generate bounded trajectories whose distributions, projected onto the first two principal components, closely match those of the true systems, while unconstrained models exhibit finite-time blow-up and yield unreliable statistics. In addition, we reported quantitative errors for the Fourier energy spectra in the Materials and Methods section.

Here, we complement those results with visualizations of the spatial Fourier energy spectra for the Lorenz 96 system and the KS reduced-order model (KS-ROM). For Lorenz--96, we treat the discrete states $x(t)\in\mathbb{R}^5$ as a periodic grid with cyclic indexing (i.e., $x_n = x_{n+5}, n\in \mathbb{Z}$) and define discrete spatial Fourier modes $\hat x_k(t)$ and the time-averaged spectrum $E(k) = \langle |\hat x_k(t)|^2\rangle_t$. For KS-ROM, the 32-dimensional state consists of the real and imaginary parts of the first $M=16$ complex Fourier coefficients $\hat u_k(t)$ of the KS solution, and we define the spectrum as $E(k) = \langle |\hat u_k(t)|^2\rangle_t$. Figures~\ref{fig:L96_spectra} and~\ref{fig:KS_spectra} compare these spectra for the ground-truth dynamics, the projected model $f^*$, and the unconstrained model $\hat f$ for Lorenz 96 and KS--ROM, respectively. In both systems, $f^*$ preserves the shape and scale of the ground-truth spectrum across the resolved modes, whereas $\hat f$ exhibits substantial spectral distortion, consistent with the quantitative spectrum errors reported in the main text.

\begin{figure}[h!]
    \centering
    \begin{subfigure}{.9\textwidth}
        \centering
        \includegraphics[width=0.9\textwidth]{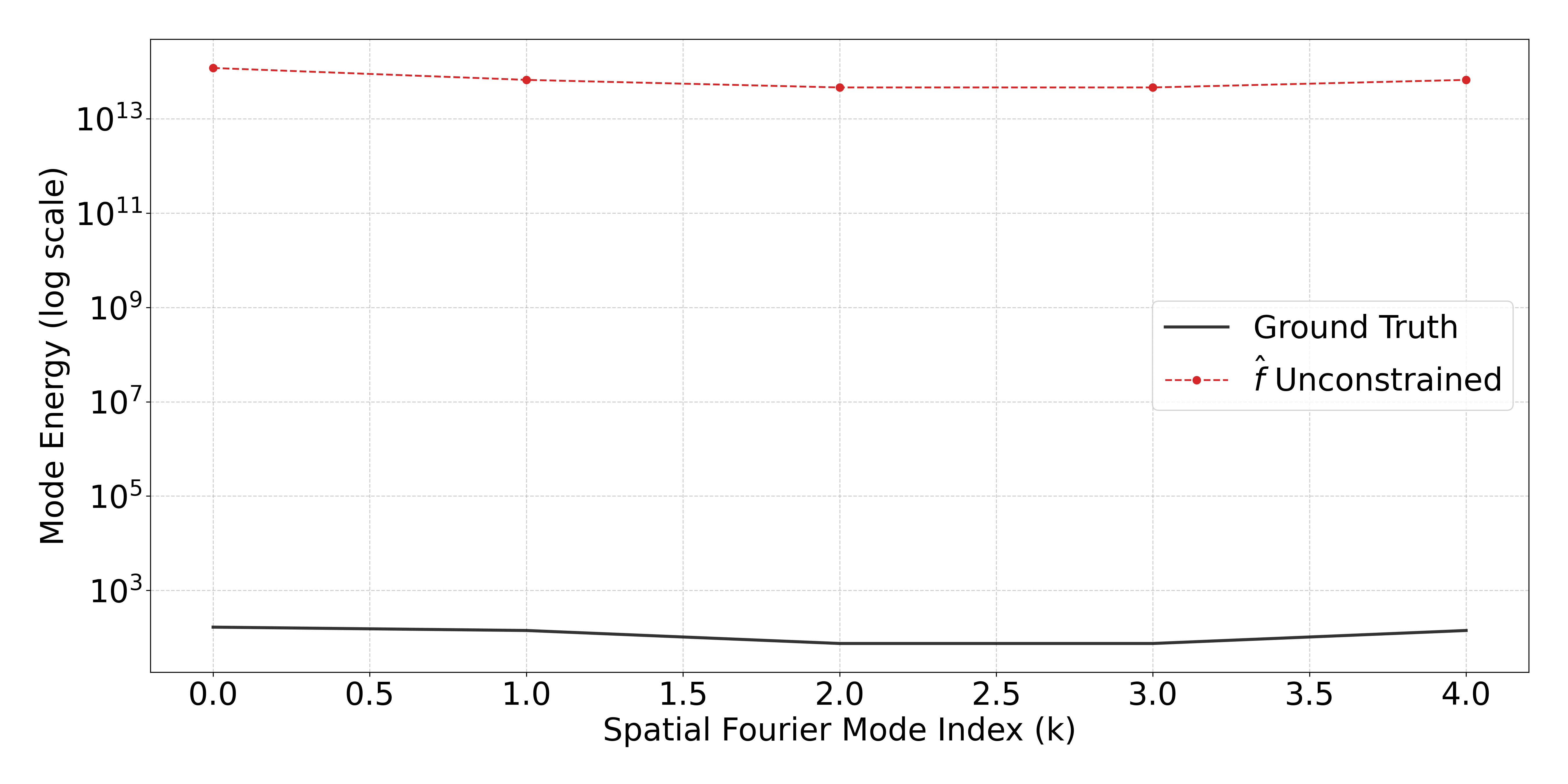}
    \end{subfigure}%
    \\
    \begin{subfigure}{.9\textwidth}
        \centering
        \includegraphics[width=0.9\textwidth]{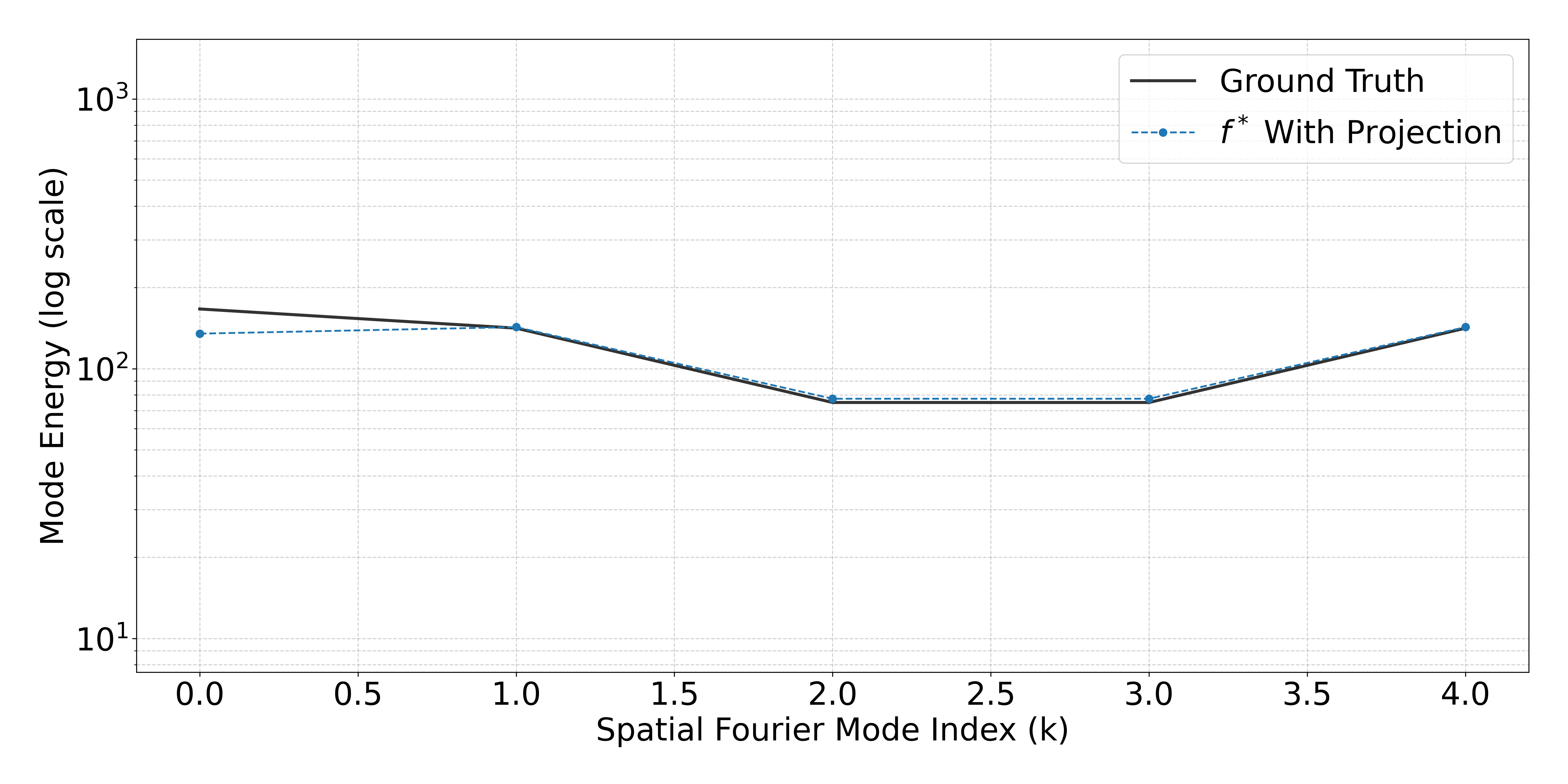}
    \end{subfigure}%
    \caption{Fourier spectrum prediction for Lorenz 96: (top) The spatial spectrum constructed from the prediction trajectory generated by the unconstrained model $\hat{f}$ (red dashed line) does not contain meaningful statistics as the prediction is orders of magnitudes higher than the true spectrum (black solid line). (bottom) With projection, the spectrum constructed from our model $f^*$ (blue dashed line) closely matches with the true spectrum (black solid line).}
    \label{fig:L96_spectra}
\end{figure}

\begin{figure}[h!]
    \centering
    \begin{subfigure}{.9\textwidth}
        \centering
        \includegraphics[width=0.9\textwidth]{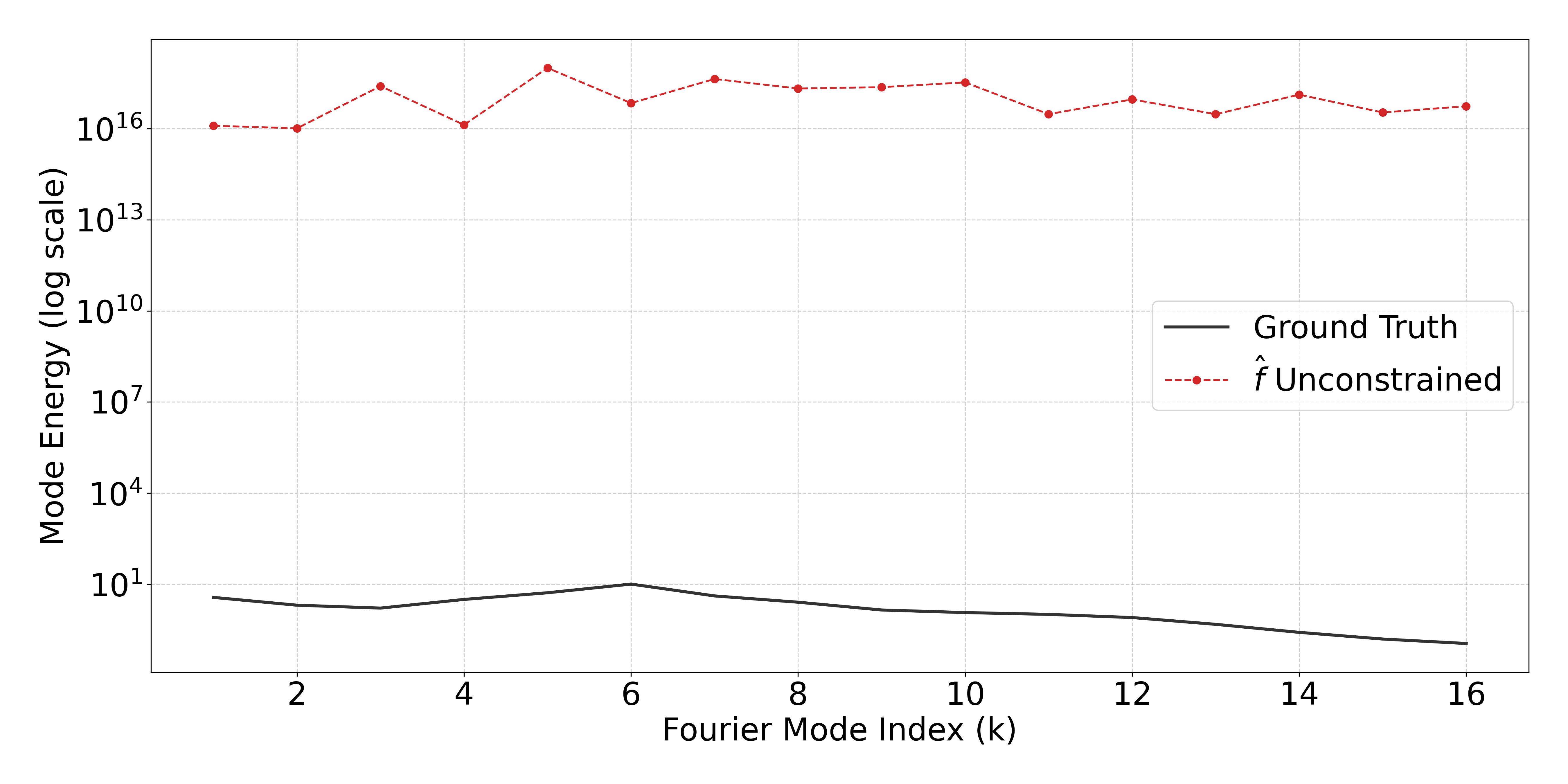}
    \end{subfigure}%
    \\
    \begin{subfigure}{.9\textwidth}
        \centering
        \includegraphics[width=0.9\textwidth]{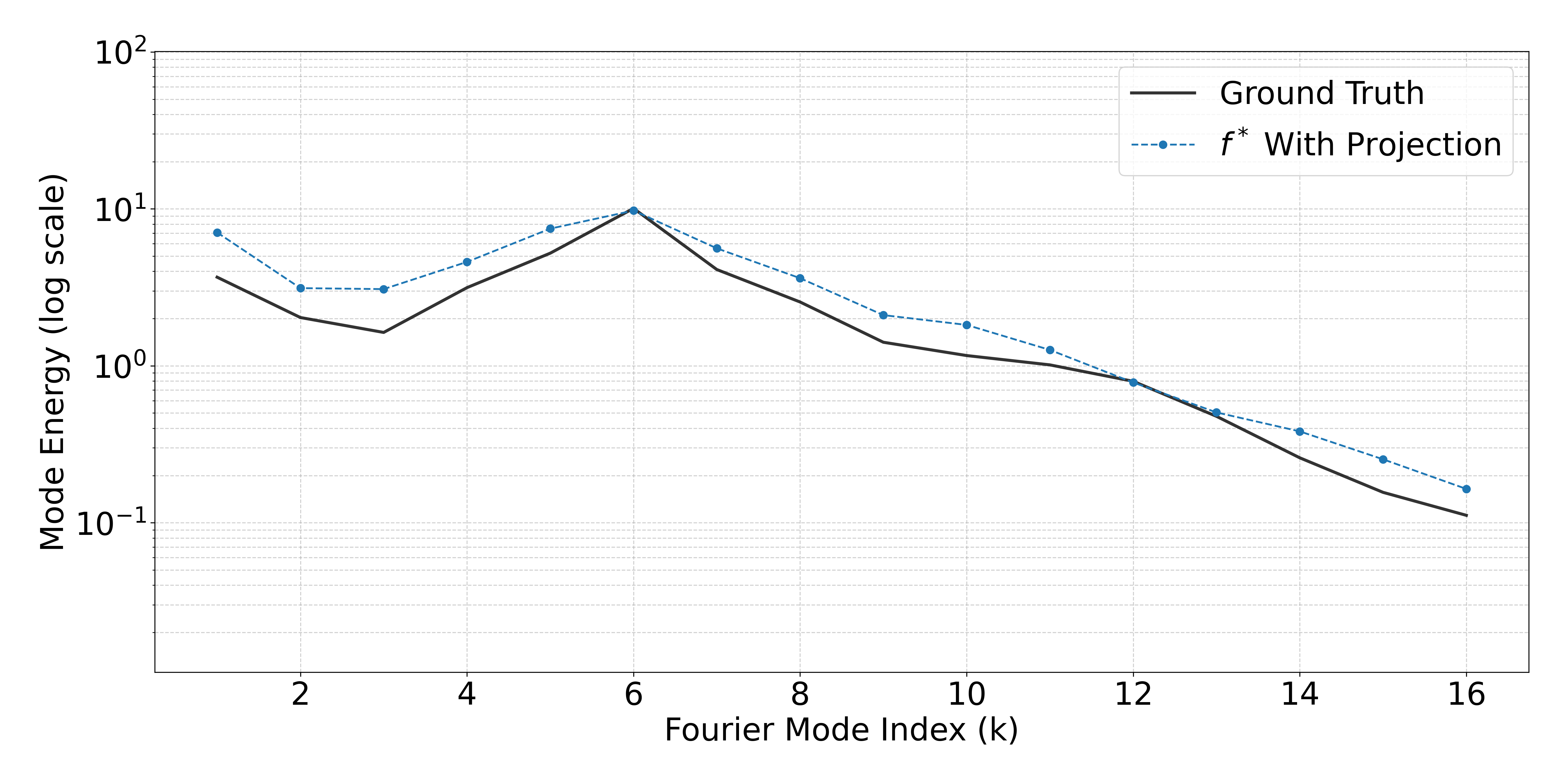}
    \end{subfigure}%
    \caption{Fourier spectrum prediction for KS--ROM: (top) The spatial spectrum constructed from the prediction trajectory generated by the unconstrained model $\hat{f}$ (red dashed line) does not contain meaningful statistics as the prediction is orders of magnitude higher than the true spectrum (black solid line). (bottom) With projection, the spectrum constructed from our model $f^*$ (blue dashed line) reproduces the energy distribution over different Fourier modes in the true system (black solid line).}
    \label{fig:KS_spectra}
\end{figure}

\clearpage
\subsection{Learned energy representation time histories}\label{sec:append_Lyap_visual}
To visualize the learned Lyapunov (energy) function and its connection to the dissipative behavior of the learned models, we evaluate a single learned energy function $V$ along trajectories generated by the true dynamics, the projected model, and the unconstrained model. Specifically, $V$ is the quadratic energy $V(x) = (x - x_0)^\top Q (x - x_0)$ learned jointly with the constrained model $f^*$. We apply this same $V$ to states from the true system $\dot{x} = f(x)$, the projected model $\dot{x}^* = f^*(x^*)$, and the unconstrained model $\dot{\hat{x}} = \hat{f}(\hat{x})$, all started from a common initial condition $x(0) = x^*(0) = \hat{x}(0)$.

In the main text we showed the energy time history for the Lorenz~63 system (Figure~3C). Here, we provide analogous visualizations for Lorenz 96 and the KS reduced-order model (KS-ROM). Figures~\ref{fig:L96_energy} and~\ref{fig:KS_energy} plot the time evolution of $V(x(t))$, $V(x^*(t))$, and $V(\hat{x}(t))$ along long rollouts of 200 seconds. In both systems, the energy trajectories under $f^*$ rapidly enter and remain within the learned invariant level set $\{x : V(x)\le c\}$ on a similar time scale as the true system, and the range of energy values attained by $f^*$ closely matches that of the ground-truth dynamics. In contrast, when the same energy function $V$ is evaluated along trajectories of the unconstrained model $\hat{f}$, the energy grows unbounded, consistent with the finite-time blow-up these models exhibit. These plots further illustrate that the learned energy function acts as a Lyapunov-like function that stabilizes the dynamics of $f^*$ and captures the overall scale of the attractors for Lorenz 96 and KS-ROM.

\begin{figure}[h!]
    \centering
    \begin{subfigure}{.9\textwidth}
        \centering
        \includegraphics[width=0.9\textwidth]{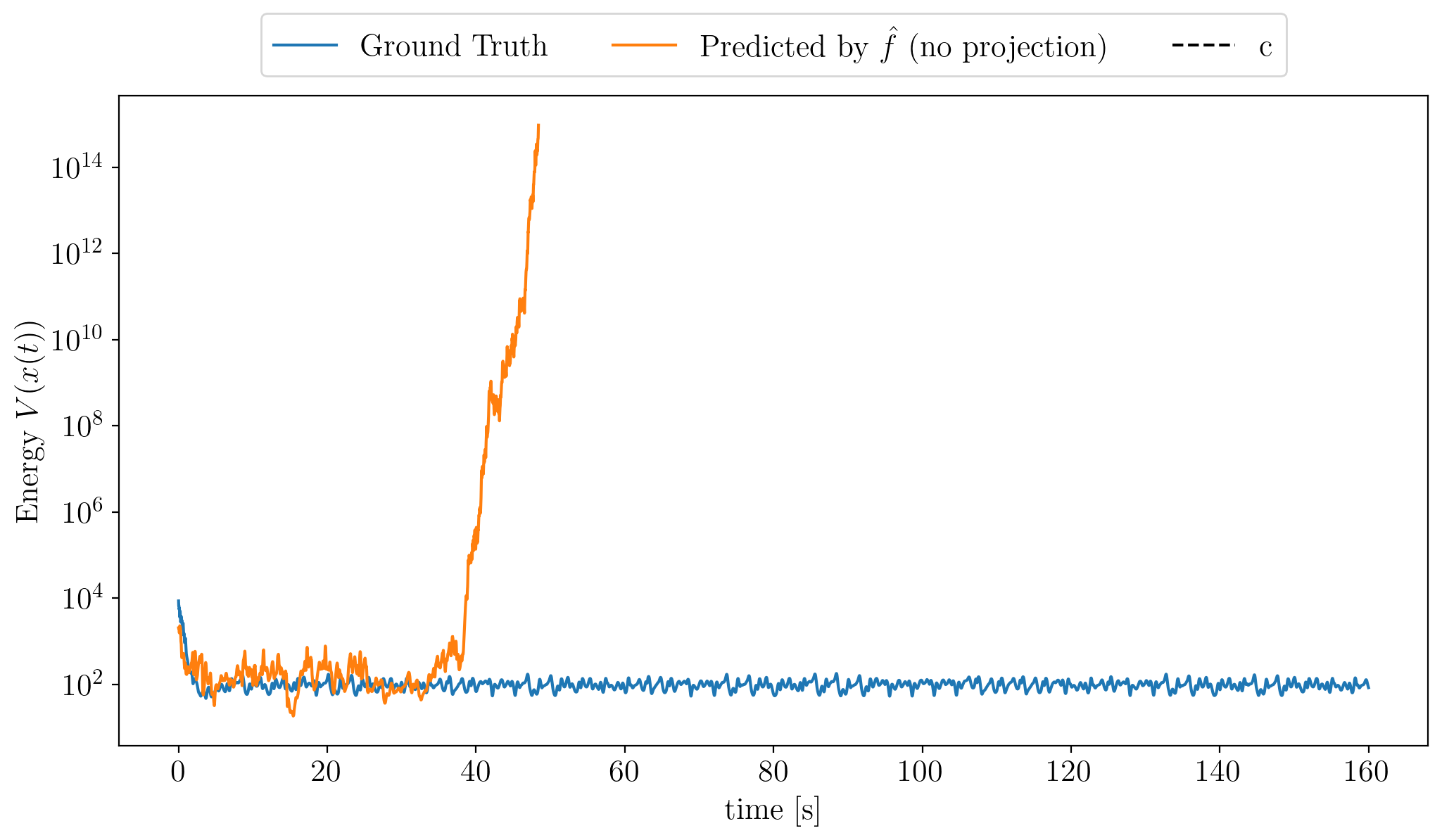}
    \end{subfigure}%
    \\
    \begin{subfigure}{.9\textwidth}
        \centering
        \includegraphics[width=0.9\textwidth]{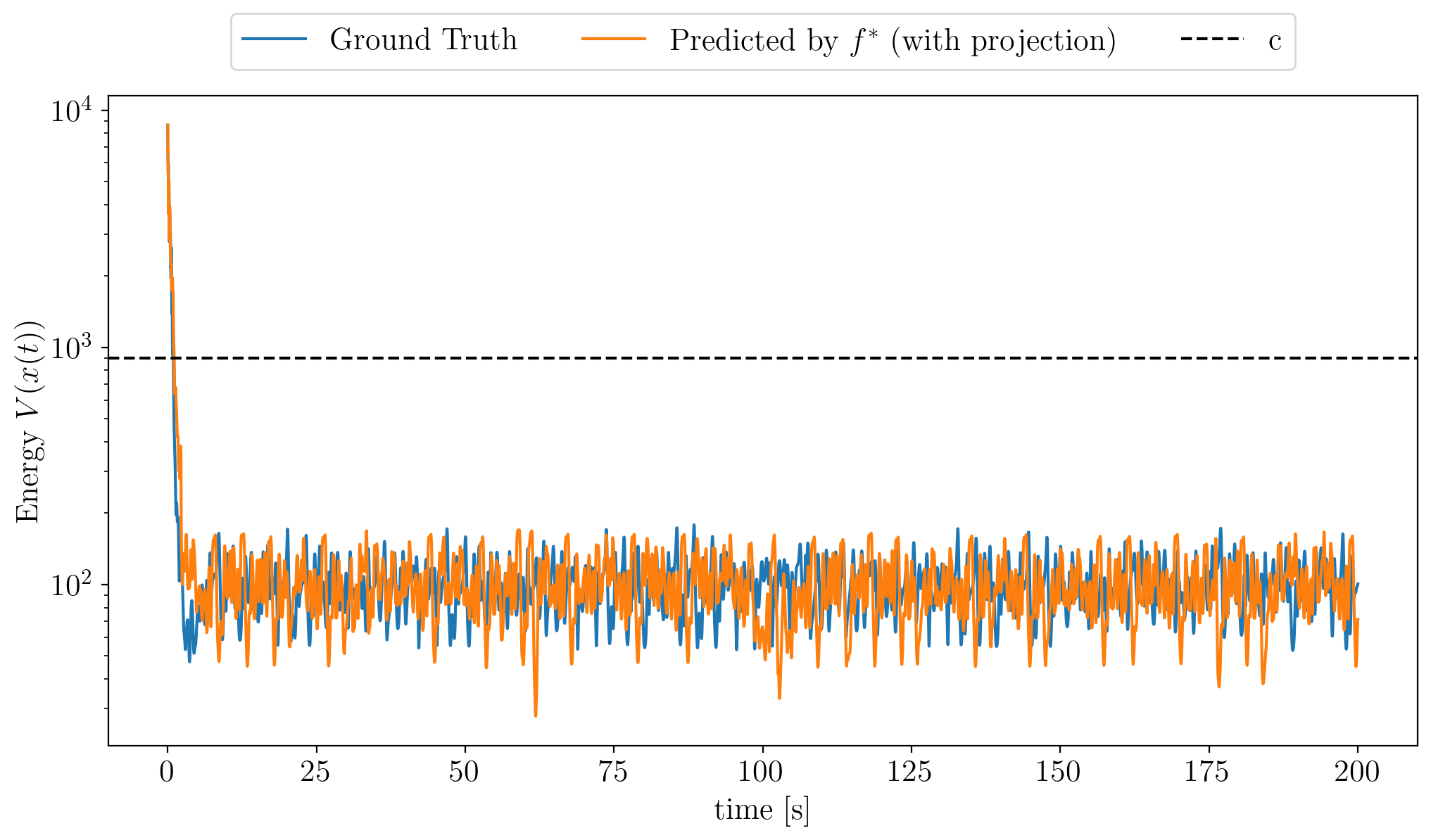}
    \end{subfigure}%
    \caption{Learned energy time histories for Lorenz–96. (Top) Energy $V(x)$ evaluated along a trajectory of the unconstrained model $\hat{f}$: the energy grows rapidly and becomes unbounded after roughly $50$ seconds, reflecting that the model exhibits finite-time blow-up. (Bottom) Energy $V(x)$ along a trajectory of the ground-truth system and the projected model $f^*$: both quickly enter and remain inside the learned invariant level set ${x : V(x)\le c}$ (black dashed line) at the same time scale, and the range of energy values attained by $f^*$ closely matches that of the true system.}
    \label{fig:L96_energy}
\end{figure}

\begin{figure}[h!]
    \centering
    \begin{subfigure}{.9\textwidth}
        \centering
        \includegraphics[width=0.9\textwidth]{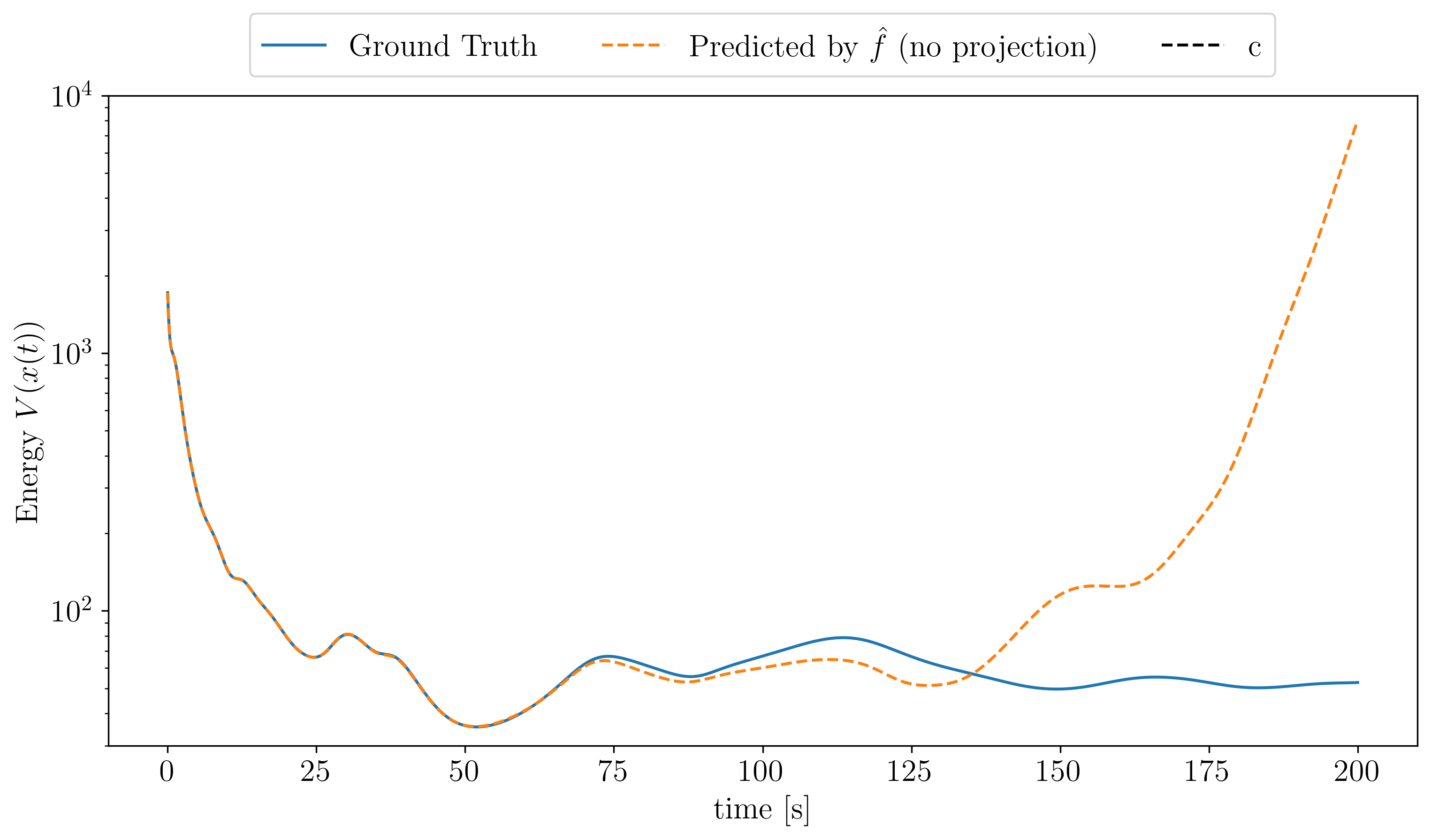}
    \end{subfigure}%
    \\
    \begin{subfigure}{.9\textwidth}
        \centering
        \includegraphics[width=0.9\textwidth]{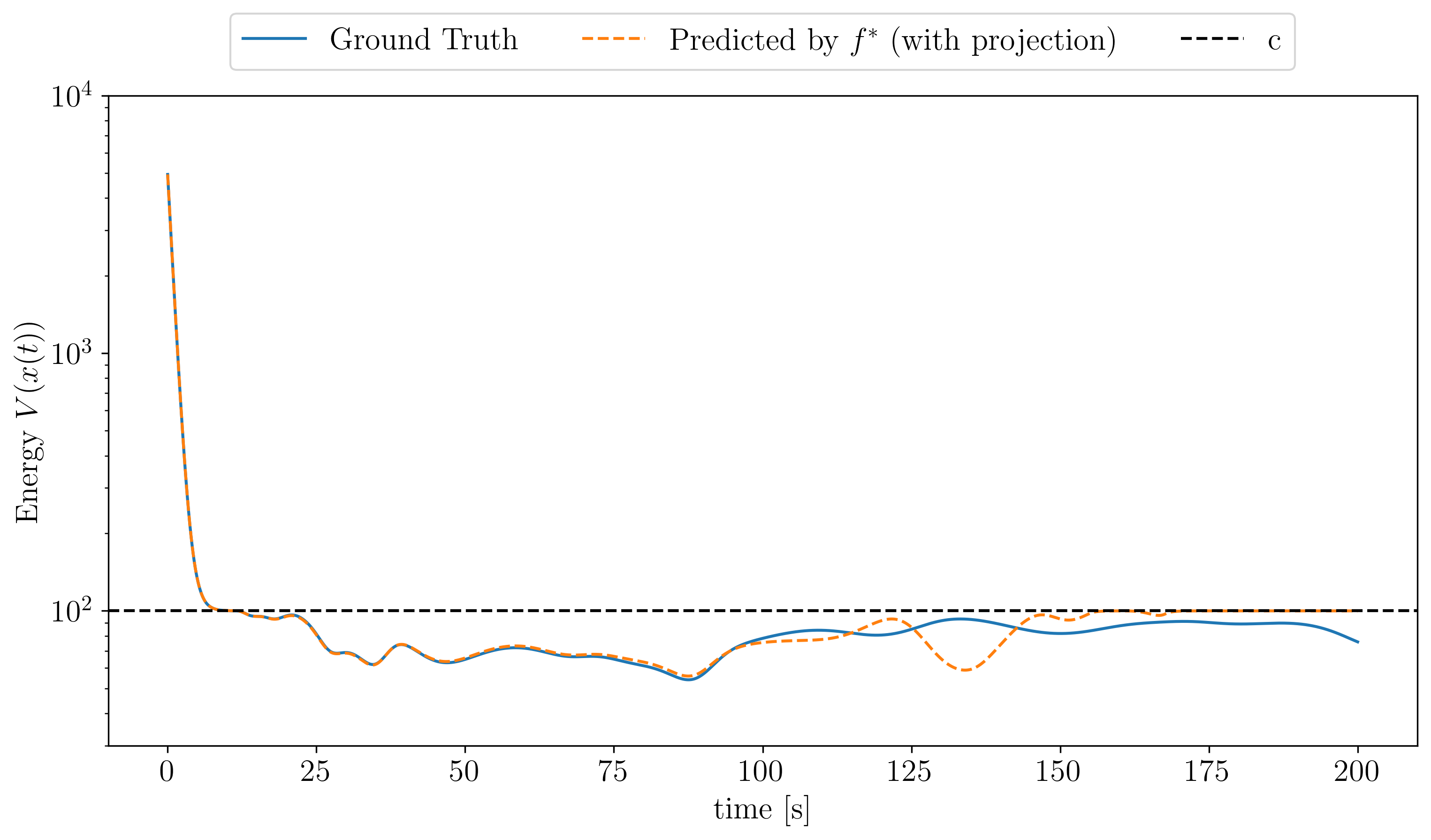}
    \end{subfigure}%
    \caption{Learned energy time histories for KS-ROM. (Top) Energy $V(x)$ along trajectories of the unconstrained model $\hat f$: after an initial period where values remain accurate, the energy grows rapidly and becomes unbounded shortly after roughly $125$ seconds, reflecting that the model exhibits finite-time blow-up. (Bottom) Energy $V(x)$ along trajectories of the ground-truth KS-ROM and the projected model $f^*$: both quickly enter and remain within the learned invariant level set ${x : V(x)\le c}$ (black dashed line), at the same time scale. Beyond roughly $125$ seconds, the trajectory under $f^*$ remains confined under the black dashed line while gradually losing pointwise accuracy, illustrating that the learned invariant set tightly bounds the true system's energy range and stabilizes the long-term dynamics even when prediction accuracy degrades.}
    \label{fig:KS_energy}
\end{figure}

\end{document}